\newif\ifFull
\Fulltrue

\ifFull
\documentclass[a4paper,11pt]{article}
\else
\documentclass[runningheads,orivec]{llncs}
\fi

\usepackage[small,labelfont=bf]{caption}

\usepackage[utf8]{inputenc}
\usepackage{cite}
\usepackage{subcaption}
\usepackage{wrapfig}

\usepackage{xcolor}

\usepackage{xspace}
\usepackage{enumerate}

\usepackage{microtype}

\usepackage{amsmath}
\ifFull
\usepackage{amsthm}
\fi
\usepackage{amssymb}
\usepackage{amstext}
\usepackage{esvect}
\usepackage{amsopn}

\usepackage{hyperref}
\usepackage{graphicx}
\usepackage[noline,boxruled, noend,algonl,boxed]{algorithm2e}
\usepackage[basic]{complexity}

\usepackage{paralist}
\defaultleftmargin{0.0em}{0.0em}{0.0em}{0.0em}
\usepackage{enumitem}

\ifFull
\usepackage[margin=1in]{geometry}
\fi

\usepackage[left, pagewise, displaymath, mathlines]{lineno}
\newcommand*\patchAmsMathEnvironmentForLineno[1]{%
  \expandafter\let\csname old#1\expandafter\endcsname\csname #1\endcsname
  \expandafter\let\csname oldend#1\expandafter\endcsname\csname end#1\endcsname
  \renewenvironment{#1}%
     {\linenomath\csname old#1\endcsname}%
     {\csname oldend#1\endcsname\endlinenomath}}%
\newcommand*\patchBothAmsMathEnvironmentsForLineno[1]{%
  \patchAmsMathEnvironmentForLineno{#1}%
  \patchAmsMathEnvironmentForLineno{#1*}}%
\AtBeginDocument{%
\patchBothAmsMathEnvironmentsForLineno{equation}%
\patchBothAmsMathEnvironmentsForLineno{align}%
\patchBothAmsMathEnvironmentsForLineno{flalign}%
\patchBothAmsMathEnvironmentsForLineno{alignat}%
\patchBothAmsMathEnvironmentsForLineno{gather}%
\patchBothAmsMathEnvironmentsForLineno{multline}%
}
\usepackage[textsize=footnotesize]{todonotes}
\SetVlineSkip{0pt}  
\def\comment#1{}%
\def\withcomments{%
  \newcounter{mycommentcounter}%
   \def\comment##1{\refstepcounter{mycommentcounter}%
    \ifhmode%
     \unskip%
     {\dimen1=\baselineskip \divide\dimen1 by 2 %
       \raise\dimen1\llap{\tiny
	{-\themycommentcounter-}}}\fi%
     \marginpar[{\renewcommand{\baselinestretch}{0.8}%
       \hspace*{-2em}\begin{minipage}{1.5\marginparwidth}\footnotesize%
[\themycommentcounter]:%
\raggedright ##1\end{minipage}}]{\renewcommand{\baselinestretch}{0.8}%
       \begin{minipage}{1.5\marginparwidth}\footnotesize%
[\themycommentcounter]: \raggedright%
##1\end{minipage}}}%
  }

\renewcommand{\C}{\ensuremath{\mathcal C}}
\newcommand{\U}{\ensuremath{\mathcal U}}
\newcommand{\A}{\ensuremath{\mathcal A}}

\newcommand{\rmin}{r_\textnormal{min}}
\newcommand{\rmax}{r_\textnormal{max}} \newcommand{\eps}{\varepsilon}
\newcommand{\wlg}{without loss of generality\xspace}

\newcommand{\roapprox}{\widetilde{r_o}}
\newcommand{\rcapprox}{\widetilde{r_c}}

\newcommand{\roexact}{\bar{r_o}}
\newcommand{\rcexact}{\bar{r_c}}

\newcounter{condition} 
\renewcommand{\thecondition}{\Roman{condition}}

\ifFull
\newcommand{\qednew}{}
\else
\newcommand{\qednew}{\hfill$\square$}
\fi

\makeatletter
\newcommand{\newcond}[1]{ \ifcsname%
  c@#1\endcsname\textnormal{(\ref{#1})}%
  \else%
  \refstepcounter{condition}\textnormal{(\thecondition)}\newcounter{#1}\label{#1}%
  \fi } \makeatother

\ifFull
\else
\spnewtheorem*{prprty}{Property}{\bfseries}{\itshape}
\fi

\ifFull
\newtheorem{theorem}{Theorem}

\newtheorem{lemma}{Lemma}

\theoremstyle{remark}
\newtheorem{property}{Property}

\else
\let\doendproof\endproof
\renewcommand\endproof{~\hfill$\qed$\doendproof}
\fi

\title{Recognizing Weighted Disk Contact Graphs}
 
\ifFull
\author{
  Boris Klemz\footnote{Institute of Computer Science, Freie Universit\"at Berlin, Germany}
  \and
  Martin N\"ollenburg\footnote{Algorithms and Complexity Group, TU Wien, Vienna, Austria}
  \and
  Roman Prutkin\footnote{Institute of Theoretical Informatics, Karlsruhe Institute of Technology, Germany}
} 
\date{}
\else 
\author{Boris Klemz\inst{1} \and Martin N\"ollenburg\inst{2} \and Roman Prutkin\inst{2}}  
\authorrunning{Klemz \and N\"ollenburg \and Prutkin}
\institute{Institute of Computer Science, Freie Universit\"at Berlin, Germany
\and
Institute of Theoretical Informatics, Karlsruhe Institute of Technology, Germany}
\fi

\begin{document}

\maketitle

\begin{abstract}
  Disk contact representations realize graphs by mapping vertices bijectively to
  interior-disjoint disks in the plane such that two disks touch each
  other if and only if the corresponding vertices are adjacent in the graph.
  Deciding whether a vertex-weighted planar graph can be realized such that the disks' 
  radii coincide with the vertex weights is known to be \NP-hard.
  In this work, we reduce the gap between hardness and tractability by
  analyzing the problem for special graph classes. We show that it
  remains \NP-hard for outerplanar graphs with unit weights and for  stars with arbitrary weights, strengthening the
  previous hardness results. On the positive side, we present constructive linear-time recognition algorithms for caterpillars with unit weights and for embedded stars with arbitrary weights.
\end{abstract} 

\section{Introduction}

A set of disks in the plane is a \emph{disk intersection representation} of a graph $G=(V,E)$ if there is a bijection between $V$ and the set of disks such that two disks intersect if and only if they are adjacent in $G$.
\emph{Disk intersection graphs} are graphs
that have a disk intersection representation; a subclass are \emph{disk contact
graphs} (also known as coin graphs), that is, graphs that have a disk intersection representation with interior-disjoint disks. This is also called a \emph{disk contact representation} (DCR) or, if connected, a circle packing. 
It is easy to see that every disk contact graph is planar and the famous Koebe-Andreev-Thurston circle packing theorem~\cite{Koebe:1936} dating back to 1936 (see Stephenson~\cite{s-cpmt-03} for its history) states that the converse is also true, that is, every planar graph is a disk contact graph.

Application areas for disk intersection/contact graphs include
modeling physical problems like wireless communication
networks~\cite{Hale:1980}, covering problems like geometric facility
location~\cite{Robert:1990, Welzl:1991}, visual representation
problems like area cartograms~\cite{Dorling:1996}
and many more (various examples are given by Clark et
al.~\cite{Clark:1990}). Efficient numerical construction of DCRs has been studied in the past~\cite{Collins:2003, Mohar:1993}. Often, however, one is interested in recognizing disk graphs or generating representations that do not only realize the
input graph, but also satisfy additional requirements. For example,
Alam et al.~\cite{Alam:2014} recently obtained several positive and negative results on the existence of balanced DCRs, in which the ratio of the largest disk radius to the
smallest is polynomial in the number of disks. Furthermore, it might
be desirable to generate a disk representation that realizes a vertex-weighted
graph such that the disks' radii or areas are proportional to the corresponding
vertex weights, for example, for value-by-area circle cartograms~\cite{Inoue:2011}.  Clearly, there exist vertex-weighted planar
graphs that cannot be realized as disk contact representations, and the
corresponding recognition problem for planar graphs is \NP-hard, even if all vertices are
weighted uniformly~\cite{Breu:1998:1}. The complexity of recognizing weighted disk contact graphs for many interesting subclasses of planar graphs remained open. Note that graphs realizable as DCRs with unit disks correspond to 1-ply graphs. This was stated by Di~Giacomo et al.~\cite{DiGiacomo:2015} who recently introduced and studied the ply number concept for graphs.
They showed that internally triangulated biconnected planar graphs admitting a DCR with unit disks can be recognized in~$O(n \log n)$ time.
In this paper we extend the results of Breu and Kirkpatrick~\cite{Breu:1998:1} and show that it remains
\NP-hard to decide whether a DCR with unit disks exists
even if the input graph is outerplanar. Our result holds both for the case that arbitrary embeddings are allowed and the case that a fixed combinatorial embedding is specified. The result for the latter case is also implied by a very recent result by Bowen et al.~\cite{Bowen:2015} stating that for fixed embeddings the problem is \NP-hard even for trees. However, the recognition of trees with a unit disk contact representation remains an interesting open problem if arbitrary embeddings are allowed. For caterpillar-trees we solve this problem in linear time. For vertex weights that are not necessarily uniform we show that the recognition problem is strongly \NP-hard even for stars if no embedding is specified. However, for embedded stars we solve the problem in linear time. Our algorithms partially use the \emph{Real RAM} model, which assumes that a set of basic arithmetic operations (including trigonometric functions and square roots) can be performed in constant time~\cite{Preparata:1985}.

\vspace{-.5ex}
\section{Unit disk contact graphs}
\label{sec:udcg}
\vspace{-1ex}

In this section we are concerned with the problem of deciding whether a given  graph is a \emph{unit disk contact graph} (UDC graph), that is, whether it has a DCR with unit disks. For a UDC graph we also say that it is \emph{UDC-realizable} or simply \emph{realizable}.
It is known since 1998 that recognizing UDC graphs is generally \NP-hard for planar graphs~\cite{Breu:1998:1}, but it remained open for which subclasses of planar graphs it can be solved efficiently and for which subclasses \NP-hardness still holds.
We show that we can recognize caterpillars that are UDC graphs in linear time and construct a representation if it exists (Section~\ref{sub:caterpillars}),  whereas the problem remains \NP-hard for outerplanar graphs (Section~\ref{sub:outerplanar}), regardless whether a combinatorial embedding must be respected or not.

\vspace{-.5ex}
\subsection{Recognizing caterpillars with a unit disk contact representation}\label{sub:caterpillars}

Let $G=(V,E)$ be a caterpillar graph, that is, a tree for which a path remains after removing all leaves. 
Let $P = (v_1, \dots, v_k)$ be this so-called \emph{inner path} of $G$.
On the one hand, it is well known that six unit disks can be tightly packed around one central unit disk, but then any two consecutive outer disks necessarily touch and form a triangle with the central disk. 
This is not permitted in a caterpillar and thus we obtain that in any realizable caterpillar the maximum degree $\Delta \le 5$.
On the other hand, it is easy to see that all caterpillars with $\Delta \le 4$ are UDC graphs as shown by the construction in Fig.~\ref{sfg:deg4cater}. 

\begin{figure}[tb]
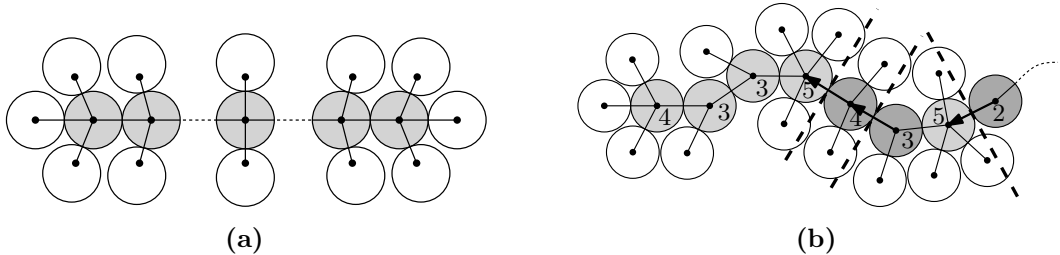

	\centering
	\subcaptionbox{\label{sfg:deg4cater}}{\includegraphics[width=0.4\columnwidth,page=2]{fig/caterpillars}}~~~~~~~~
	\subcaptionbox{\label{sfg:incrementalcater}}{\includegraphics[width=0.4\columnwidth,page=3]{fig/caterpillars}}
	\vspace{-1ex}
	\caption{(a) For $\Delta \le 4$ any caterpillar can be realized. (b) Incremental construction of a DCR. Narrow disks are dark gray  and indicated by an outgoing arrow, wide disks are light gray. }
	\vspace{-2ex}
	\label{fig:caterpillar} 
\end{figure}

However, not all caterpillars with $\Delta = 5$ can be realized.
For example, two degree-5 vertices on~$P$ separated by zero or more degree-4 vertices cannot be realized, as they would again require tightly packed disks inducing cycles in the contact graph. In fact, we get the following characterization.

\begin{lemma}\label{lem:cater}
	A caterpillar $G$ with $\Delta = 5$ is a UDC graph if and only if there is at least one vertex of degree at most 3 between any two degree-5 vertices on the inner path $P$.
\end{lemma}

\begin{proof}
	Consider an arbitrary UDC representation of $G$ and let $D_i$ be the disk representing vertex $v_i$ of the inner path $P$. %
	Let $\ell_i$ be the tangent line between two adjacent disks $D_{i-1}$ and $D_i$ on $P$. 
	We say that $P$ is \emph{narrow} at $v_i$ if some leaf disk attached to $D_{i-1}$ intersects $\ell_i$; otherwise $P$ is \emph{wide} at $v_i$. 
	Let $v_i$ and $v_j$ ($i<j$) be two degree-5 vertices on $P$ with no other degree~5 vertices between them.	
	The path $P$ must be narrow at the next vertex $v_{i+1}$, since one of the four mutually disjoint neighbor disks of $D_{i-1}$ except $D_i$ necessarily intersects $\ell_i$. 
	If there is no vertex $v_k$ $(i<k<j)$ with $\deg(v_k) \le 3$ between $v_i$ and $v_j$ we claim that $P$ is still narrow at $v_j$. 
	If $j=i+1$ this is obviously true. 
	Otherwise all vertices between $v_i$ and $v_j$ have degree 4.
	But since the line $\ell_{i+1}$ was intersected by a neighbor of $v_i$, this property is inherited for the line $\ell_{i+2}$ and a neighbor of $v_{i+1}$ if $\deg(v_{i+1})=4$.
	An inductive argument applies.
	Since $P$ is still narrow at the degree-5 vertex $v_j$, it is impossible to place four mutually disjoint disks touching $D_{j}$ for the neighbors of $v_j$ except $v_{j-1}$. 
	
	We now construct a UDC representation for a caterpillar in which any two degree-5 vertices of $P$ are separated by a vertex of degree $\le 3$.
	We place a disk $D_1$ for $v_1$ at the origin and attach its leaf disks \emph{leftmost}, that is, symmetrically pushed to the left with a sufficiently small distance between them. 
	In each subsequent step, we place the next disk $D_i$ for $v_i$ on the bisector of the \emph{free space}, which we define as the maximum cone with origin in $D_{i-1}$'s center containing no previously inserted neighbors of $D_{i-1}$ or $D_{i-2}$. Again, we attach the leaves of~$D_{i}$ in a leftmost and balanced way, see Fig.~\ref{sfg:incrementalcater}.
	For odd-degree vertices this leads to a change in direction of $P$, but by alternating upward and downward bends for subsequent odd-degree vertices we can maintain a horizontal monotonicity, which ensures that leaves of $D_i$ can only collide with leaves of $D_{i-1}$ or $D_{i-2}$.
	In this construction $P$ is wide until the first degree-5 vertex is placed, after which it gets and stays narrow as long as degree-4 vertices are encountered.
	But as soon as a vertex of degree $\le 3$ is placed, $P$ gets (and remains) wide again until the next degree-5 vertex is placed. 
	Placing a degree-5 vertex at which $P$ is wide can always be done. \qednew
\end{proof}

Lemma~\ref{lem:cater} and the immediate observations for caterpillars with $\Delta \ne 5$ yield the following theorem. We note that the decision is only based on the vertex degrees in $G$, whereas the construction uses a Real RAM model.

\begin{theorem}
  \label{thm:ud-caterpillars}
  For a caterpillar $G$ it can be decided in linear time whether $G$ is a UDC graph if arbitrary embeddings are allowed. 
  A UDC representation (if one exists) can be constructed in linear time.
\end{theorem}

\subsection{Hardness for outerplanar graphs} \label{sub:outerplanar}

A planar 3SAT formula $\varphi$ is a Boolean 3SAT formula with a set $\U$ of variables and a set $\C$ of clauses such that its \emph{variable-clause-graph} $G_\varphi = (\U \cup \C, E)$ is planar. The set $E$ contains for each clause $c \in \C$ the edge $(c,x)$ if a literal of variable $x$ occurs in $c$. Deciding the satisfiability of a planar 3SAT formula is \NP-complete~\cite{Lichtenstein:1982} and there exists a planar drawing $\mathcal{G}_\varphi$ of $G_\varphi$ on a grid of polynomial size such that the variable vertices are placed on a horizontal line and the clauses are connected in a comb-shaped rectangular fashion from above or below that line~\cite{Knuth:1992}, see Fig.~\ref{fig:highLevelGrid}a. A planar 3SAT formula $\varphi$ is \emph{monotone} if each clause contains either only positive or only negative literals and if $G_\varphi$ has a planar drawing as described before with all clauses of positive literals on one side and all clauses of negative variables on the other side. The 3SAT problem remains \NP-complete for planar monotone formulae~\cite{Lichtenstein:1982} and is called Planar Monotone 3-Satisfiability (PM3SAT). 

\begin{figure}[tb]
\centering
\includegraphics[width=0.99\columnwidth]{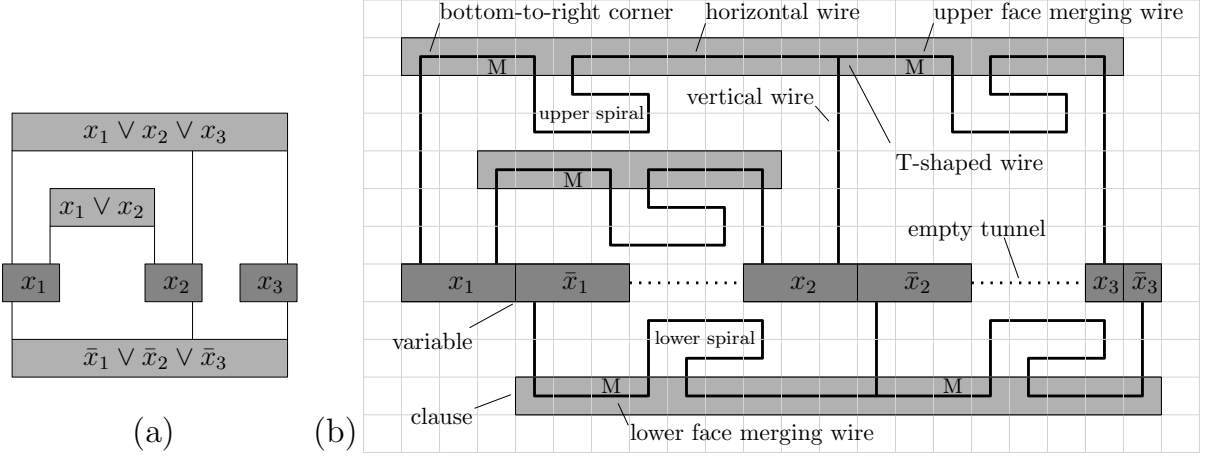}
\caption{Sketch of the grid layout $\mathcal{G}_\varphi$ (a) and high-level structure of the construction of $G'_\varphi$ (b) for the PM3SAT formula $\varphi = (x_1\vee x_2)\wedge (x_1\vee x_2\vee x_3)\wedge (\bar x_1\vee \bar x_2\vee \bar x_3)$.}
\label{fig:highLevelGrid}
\end{figure}

We perform a polynomial reduction from PM3SAT to show \NP-hardness of recognizing (embedded) outerplanar UDC graphs. 
A graph is \emph{outerplanar} if it has a planar drawing in which all vertices lie on the unbounded outer face.
We say that a planar graph $G$ is (combinatorially) \emph{embedded} if we are given for each vertex the circular order of all incident edges as well as the outer face such that a planar drawing respecting this embedding exists. %
For the reduction we create, based on the planar drawing $\mathcal{G}_\varphi$, an (embedded) outerplanar graph $G'_\varphi$ that has a UDC representation if and only if the formula $\varphi$ is satisfiable.

Arguing about UDC representations of certain subgraphs of $G'_\varphi$ becomes a lot
easier, if there is a single unique geometric representation (up
to rotation, translation and mirroring). We call graphs with such a representation
\emph{rigid}. In the following Lemma we state a sufficient condition for rigid UDC structures. Note that all subgraphs of $G'_\varphi$ that we refer to as rigid satisfy this condition.

\newcommand{\LemRigidText}{%
Let $G=(V,E)$ be a biconnected graph realizable as a UDC representation that induces an internally triangulated outerplane embedding of~$G$. Then, $G$ is rigid.   
}
\begin{lemma}\label{lem:rigid}
  \LemRigidText	
\end{lemma}
\begin{proof}
Let $\mathcal G$ be a UDC representation of $G$ that induces an internally triangulated outerplane embedding of $G$. We show by induction that our hypothesis is true for any natural number~$n=|V|$ of vertices. For the induction base case we consider~$1\le n\le 3$. If~$n=1$, then~$G$ is obviously rigid. Since~$G$ is biconnected we know that $n\neq 2$. If~$n=3$, then~$G$ is a complete graph since $G$ is biconnected. In this case~$G$ is obviously rigid, which concludes the induction base case.

For the induction step, consider any~$n>3$ and assume that our hypothesis holds true for all graphs with at most~$n-1$ vertices. By assumption~$G$ is outerplanar and thus there exists a vertex~$v_r\in V$ with~$\mathrm{deg}(v_r)\le 2$. 
Since~$G$ is also biconnected and~$n>3$ we specifically know that $\mathrm{deg}(v_r)=2$. Let~$v_1,v_2\in V$ be the neighbors of~$v_r$. Removing the disk corresponding to~$v_r$ from~$\mathcal G$ yields a UDC representation~$\mathcal G'$ that realizes the subgraph~$G'=(V',E')$ of~$G$ that is induced by the vertex set~$V'=V\setminus \lbrace v_r\rbrace $. The induced planar drawing of~$\mathcal G'$ is obviously still outerplanar and internally triangulated.

\begin{figure}
  \centering
   \includegraphics[width=0.3\textwidth]{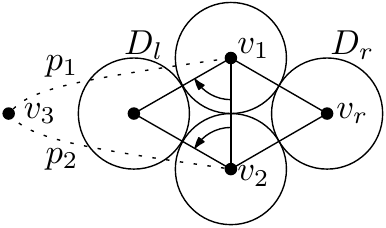}
   \caption{Usage of Menger's Theorem in Lemma~\ref{lem:rigid}.}
  \label{fig:UDTrigid03}
\end{figure}

The following steps are illustrated in Fig.~\ref{fig:UDTrigid03}. The number of vertices of~$G$ is~$n>3$ and~$G$ is biconnected. By Menger's Theorem there exist two internally vertex-disjoint paths~$p_1,p_2$ (via~$v_1$ and~$v_2$ respectively) between~$v_r$ and some vertex~$v_3\in V'$ with~$v_3\neq v_r,v_1,v_2$. The existence of~$p_1$ and~$p_2$ together with~$\mathrm{deg}(v_r)=2$ and with the fact the induced drawing of~$\mathcal G'$ is internally triangulated imply that~$e=\lbrace v_1,v_2\rbrace \in E'$. The existence of~$e$, on the other hand, implies that~$G'$ is biconnected since~$G$ is biconnected and since~$|V'|\ge 3$.~$G'$ and~$\mathcal G'$, therefore, meet all of our preconditions and by our induction hypothesis we know that~$G'$ is rigid since~$|V'|=n-1$.

We now re-insert the vertex~$v_r$ and edges~$\lbrace v_r,v_1\rbrace,\lbrace v_r,v_2\rbrace $ to~$G'$ (resulting in~$G$) and add a corresponding disk~$D_r$ to~$\mathcal G'$. In the following paragraphs we argue that there exists exactly one location where we can place~$D_r$ (namely the same location used in~$\mathcal G$) implying that~$G$ is also rigid since the obtained representation is then congruent to~$\mathcal G$, which concludes the induction step and the proof.

Due to the existence of~$e$, we know that the two disks~$D_1$ and~$D_2$ that correspond to~$v_1$ and~$v_2$ respectively touch each other. Hence, there exist exactly two locations for disks that correspond to common neighbors of both~$v_1$ and~$v_2$. One of these locations is the location of the disk corresponding to~$v_r$ in~$\mathcal G$. It suffices to show that in both~$\mathcal G$ and~$\mathcal G'$ there exists a disk~$D_l$ that occupies the other possible location. 

The existence of~$p_1$ and~$p_2$ implies that the degree of both~$v_1$ and~$v_2$ is at least~$3$ and that there exists a path~$p_3$ between~$v_1$ and~$v_2$ that contains~$v_3$ and that does not contain~$v_r$ or~$e$. Assume, without loss of generality, that in the drawing induced by~$\mathcal G$ vertex~$v_r$ is the neighbor of~$v_1$ that clockwise precedes~$v_2$ and that~$v_r$ is the neighbor of~$v_2$ that clockwise succeeds~$v_1$. Let~$v_1'$ be the neighbor of~$v_1$ that clockwise succeeds~$v_2$ and let~$v_2'$ be the neighbor of~$v_2$ that clockwise precedes~$v_1$. By the fact that the drawing of~$G$ induced by~$\mathcal G$ is internally triangulated and that~$v_r$ has to be adjacent to the outer face of the drawing, we can conclude that~$v_1'=v_2'$, which, therefore, is another common neighbor of~$v_1$ and~$v_2$ whose corresponding disk has to be~$D_l$. \qednew
\end{proof}

The main building block of the reduction is a \emph{wire gadget} in  $G'_\varphi$ that comes in different variations but always consists of a rigid tunnel structure containing a rigid bar that can be flipped into different tunnels around its centrally located articulation vertex. 
Each wire gadget occupies a square tile of fixed dimensions so that different tiles can be flexibly put together in a grid-like fashion. 
The bars stick out of the tiles in order to transfer information to the neighboring tiles. 
Variable gadgets consist of special tiles containing tunnels without bars or with very long bars. Adjacent variable gadgets are connected by narrow tunnels without bars. 
\emph{Face merging} wires work essentially like normal horizontal wires but their low-level construction differs in order to assert that $G'_\varphi$ is outerplanar and connected.
Figure~\ref{fig:highLevelGrid}b shows a schematic view of how the gadget tiles are arranged to mimic the layout $\mathcal G_\varphi$ of Fig.~\ref{fig:highLevelGrid}a. The wires connect the positive (negative) clauses to the left (right) halves of the respective variable gadgets. Furthermore, we place a face merging wire (marked by `M') in the top/bottom left corner of each inner face followed by an \emph{upper (lower) spiral}, which is a fixed $3\times 4$ pattern of wire gadgets. These structures ensure that $G'_\varphi$ is outerplanar and they limit relative displacements.

\begin{wrapfigure}[16]{r}{.31\textwidth}
    \vspace{-2ex}
	\centering
	\includegraphics[scale=.35]{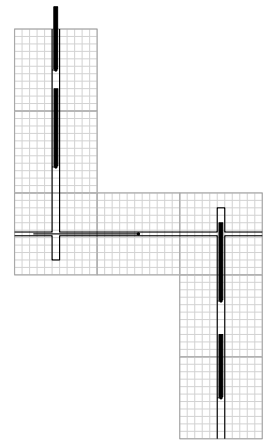}
	\caption{Variable gadget in  state \emph{false} with a positive (left) and a negative literal (right).}
	\label{fig:UDCvariable}
\end{wrapfigure}
The main idea behind the reduction is as follows. 
Each \emph{variable gadget} contains one thin, long horizontal bar that is either flipped to the left (\emph{false}) or to the right (\emph{true}), see Fig.~\ref{fig:UDCvariable}. 
If the bar is in its left (right) position, this blocks the lower (upper) bar position of the first wire gadget of each positive (negative) literal. 
Consequently, each wire gadget that is part of the connection between a variable gadget and a clause gadget must flip its entire chain of bars towards the clause if the literal is false.  
The design of the \emph{clause gadget} depends on its number of literals. Figure~\ref{fig:clause-T} illustrates the most important case of a clause with three literals containing a T-shaped wire gadget. 
The bar of the T-shaped wire needs to be placed in one of the three incident tunnels. This is possible if and only if at least one of the literals evaluates to \emph{true}. A similar  statement holds true for clauses with two or one literals; their construction is much simpler: just a horizontal wire gadget or a dead end suffice as clause tile.

\begin{figure}[b]
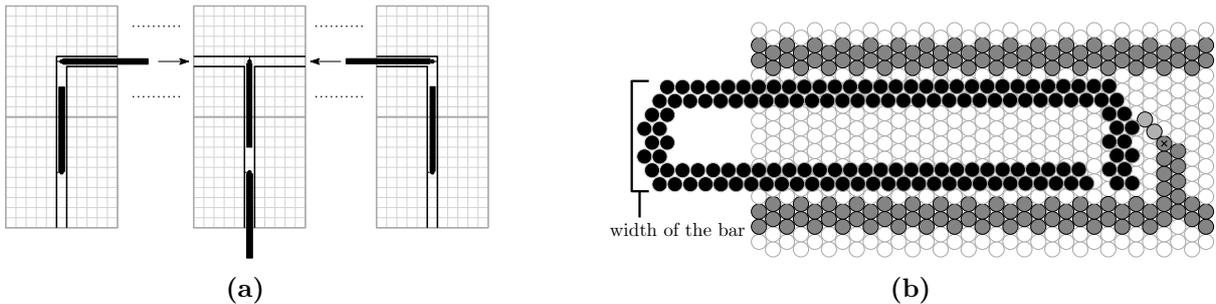

	\centering
	\subcaptionbox{\label{fig:clause-T}}{\includegraphics[width=.4\textwidth]{fig/clause-T}}
	\hfill
	\subcaptionbox{\label{fig:hor-wire}}{\includegraphics[width=.5\textwidth]{fig/hor-wire-detail}}
	\caption{(a) Clause gadget with two false inputs (left and right) and one true input. (b) Detailed view of a horizontal wire gadget with a rigid bar (black disks) inside a tunnel (dark gray disks).}
	\label{fig:gadgets}
\end{figure}

All gadgets are realized by combining several rigid UDC subgraphs. 
As an example, Fig.~\ref{fig:hor-wire} shows a close-up of the left side of a horizontal wire gadget. 
Both the black and the dark gray disks form rigid components whose UDC graphs satisfy the precondition of Lemma~\ref{lem:rigid}.
The black disks implement the bar, the dark gray disks constitute the tunnel.
Note how the bar can be flipped or mirrored to the left or the right around the articulation disk (marked `x') due to the two light gray disks (called \emph{chain} disks) that do not belong to a rigid structure.
The width of each bar is chosen such that it differs from the supposed inner width of a tunnel by at most twice the disk diameter, thus admitting some slack.
However, we can choose the width of the tunnels/bars (and the gadget tile dimensions) as a large polynomial in the input such that this ``wiggle room" does not affect the combinatorial properties of our construction. The description of the face merging wire below discusses this aspect in more detail.
Further, we choose the lengths of the bars such that the bars of two adjacent wire gadgets collide if their bars are oriented towards each other.
Figure~\ref{fig:Twire} depicts a close-up of a T-shaped wire used in the clause gadgets.
\begin{figure}
  \centering
   \includegraphics[width=.95\textwidth]{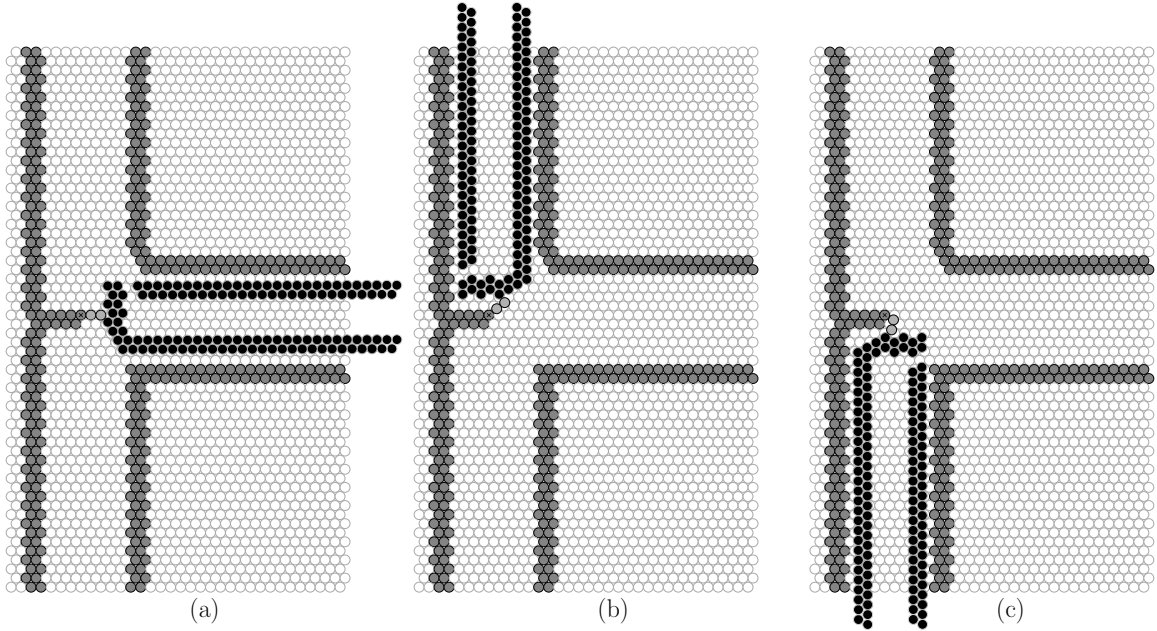}
   \caption{The three possible orientations of the bar in a T-shaped wire (turned by $90^\circ$ for space reasons).}
  \label{fig:Twire}
\end{figure}
Unlike the bars of wire gadgets, the bars of variable gadgets are not designed to transmit information from tile to tile.
Instead they are simply designed to prevent the adjacent vertical wires on either the left or the right side of the variable gadget to be oriented towards it.
For this reason, we can choose the width of the variable bars to be very small (e.g., just $2$ disks), in order to obtain an overall tighter construction.

\begin{wrapfigure}[15]{r}{.45\textwidth}
	\centering
	\includegraphics[width=.44\textwidth]{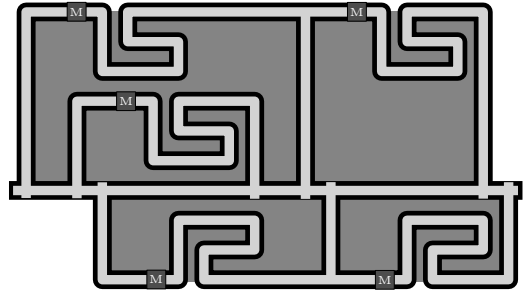}
	\caption{Schematic of $G'_\varphi$ if face merging wires (marked  `M') replace some regular wires. Inner faces of $\mathcal G_\varphi$ in dark gray, the face `inside the tunnels' in light gray, the outer face in white and the face boundaries in black.}
	\label{fig:overview-face-mergers}
\end{wrapfigure}

Now that we have established how the gadgets work and how they are constructed, consider the properties of the corresponding graph $G'_\varphi$ that encodes the entire structure.
If we would use only the regular wire gadgets as in Fig.~\ref{fig:hor-wire} for the entire construction, $G'_\varphi$ would neither be outerplanar nor connected.
As illustrated in Fig.~\ref{fig:overview-face-mergers}, %
for each of the inner faces of  $\mathcal G_\varphi$ we would obtain a single rigid structure, which we call \emph{face boundary}, with several bars attached to it. These face boundaries, however, would not be connected to each other. Furthermore, the subgraphs that realize the face boundaries would not be outerplanar.
This is why we replace some horizontal wire gadgets in the upper (positive) and lower (negative) part of our construction by \emph{upper} and \emph{lower} face merging wires respectively, which have two purposes. Horizontal wires contain a tunnel that is formed by two face boundaries, called the \emph{upper} and \emph{lower} face boundary of the corresponding gadget tile. These face boundaries are not connected, see Figure~\ref{fig:hor-wire}. In a face merging wire, however, the respective face boundaries are connected. Furthermore, a gap is introduced (by removing two disks) to the lower (upper) face boundary in an upper (lower) face merging wire so that the lower (upper) face boundary now becomes outerplanar.
Since the face merging wire is supposed to transfer information just like a  horizontal wire we cannot connect the two face boundaries rigidly.
Instead we create three bars connected to each other with chain disks, see Fig.~\ref{fig:maintenance}a.
The width of the top and the bottom bars are chosen such that they fit tightly inside the narrow cavity in the middle of the tile if placed perpendicularly to the left or right of the respective articulation disk.
The third bar ensures that all three bars together are placed either to the left (Fig.~\ref{fig:maintenance}b) or to the right (Fig.~\ref{fig:maintenance}a), which allows the desired information transfer.

Together with the incident spiral, a face merging wire ensures that the disks of the lower face boundary deviate from their intended locations relative to the upper face boundary only by up to a small constant distance since (1) the design and the asymmetrical placement of the spirals and the face merging wires preserve the orientations of the respective upper and lower face boundaries, i.e., the left/right/top/bottom sides of these structures are facing as intended in any realization and (2) the width of the tunnels is at most twice the disk diameter larger than the width of the bars and there is at least one bar located in any of the cardinal directions of each spiral.
This effect can cascade since the face boundaries might be connected to further face boundaries.
However, according to Euler's formula the number of faces in $\mathcal G_\varphi$ is linear in the number of clauses and variables and, therefore, the total distance by which a disk can deviate from its intended ideal position is also linear in this number.
By accordingly adjusting the tile dimensions and bar widths, we can therefore ensure that the wiggle room in our construction does not affect the intended combinatorial properties while keeping the size of $G'_\varphi$ polynomial.
The introduction of face merging wires causes $G'_\varphi$ to be connected and it causes all inner faces and the face 'inside the tunnels' to collapse.
Finally, by introducing a single gap in the outermost rigid structure, $G'_\varphi$ becomes outerplanar, which concludes our reduction.

This concludes our construction for the case with arbitrary embeddings. Note, however, that the gadgets are designed such that flipping the bars does not
 require altering the combinatorial embedding of the graph.
This holds true even for the face merging wire.
Therefore, we can
easily provide a combinatorial embedding such that $G'_\varphi$
can be realized with respect to said embedding if and only if
$\varphi$ is satisfiable.
Thus, we obtain the following theorem.

\begin{figure}[tb]
	\centering
		\includegraphics[width=.95\columnwidth]{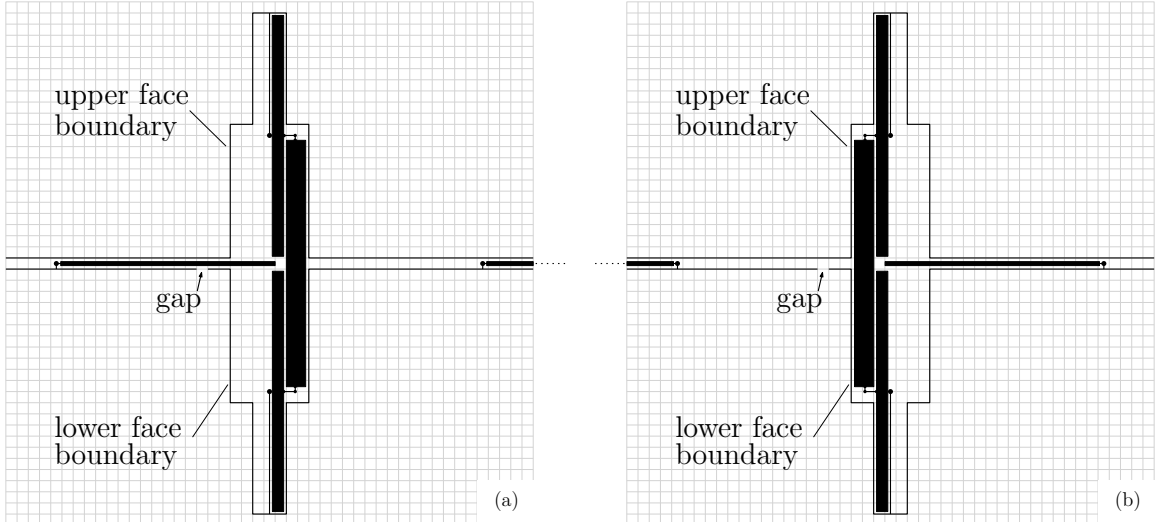}
	\caption{Upper face merging wire gadget oriented to the right (a) / left (b). It connects the lower and upper face boundaries. The gap causes the faces inside the tunnel and the lower face to collapse.}
	\label{fig:maintenance}\vspace{-2ex}
\end{figure}

\newcommand{\ThmUdOuterplanarHard}{For outerplanar graphs the UDC recognition problem is \NP-hard if an arbitrary embedding is allowed. The same holds for outerplanar graphs with a specified combinatorial embedding.}
\begin{theorem}
  \label{thm:ud-outerplanar-hard}
  \ThmUdOuterplanarHard
\end{theorem}
\begin{proof}
If the PM3SAT formula $\varphi =(\mathcal U,\mathcal C)$ is satisfiable, there exists a truth assignment $t$ for $\mathcal U$ that satisfies all clauses in $\mathcal C$, i.e., one literal of each clause in $\mathcal C$ is true with respect to $t$. Each literal of a clause $c\in \mathcal C$ corresponds to one of the vertical wires incident to the clause gadget representing $c$. A truth assignment for $\mathcal U$ \emph{induces} an orientation for each of these wires. We orient the bars of true literals towards the respective variable gadget and the false literals towards the clause gadget. Thus, $t$ induces an orientation in which at least one of $c$'s bars is oriented towards a variable gadget, which is a necessary and sufficient condition for the realizability of the clause gadget subgraph of~$c$. Furthermore, a truth assignment \emph{induces} an orientation for the bars in the variable gadgets. We orient such a bar to the right (left) if the corresponding variable is true (false). Recall that all clauses above (below) the horizontal line of variables contain exclusively positive (negative) literals and that the successions of wires that represent these literals are connected to the left (right) side of a variable gadget. Wires of a true literal are oriented towards the respective variable gadget, thus, the bar of the vertical wire incident to the variable gadget sticks into variable gadget. This does not interfere with the realizability of the variable gadgets' subgraphs due to the fact that the orientation of the variable gadgets' bars induced by $t$ is chosen such that these bars only prevent false literal wires to be oriented towards variable gadgets. The false literal wires have to be oriented towards their respective clauses in accordance with the orientation (induced by $t$) of the vertical wires incident to the clause gadgets. Thus, if $\varphi$ is satisfiable, $G'_\varphi$ is realizable. On the other hand, if $G'_\varphi$ is realizable, the orientation of the variable gadgets' bars induces a satisfying truth assignment for $\varphi$ so that $\varphi$ is satisfiable if and only if $G'_\varphi$ is realizable. \qednew
\end{proof}

Note that the \NP-hardness of UDC recognition for \emph{embedded} outerplanar graphs is implied by the recent result of Bowen et al.~\cite{Bowen:2015} showing the \NP-hardness of UDC recognition for embedded trees. It is, however, still open whether Theorem~\ref{thm:ud-outerplanar-hard} extends to trees without a fixed embedding.

\section{Weighted disk contact graphs}

In this section, we assume that a positive weight $w(v)$ is assigned to each vertex $v$ of the graph $G=(V,E)$. The task is to decide whether $G$ has a DCR, in which each disk $D_v$ representing a vertex $v \in V$ has radius proportional to $w(v)$. A DCR with this property is called a \emph{weighted disk contact representation} (WDC representation) and a graph that has a WDC representation is called a \emph{weighted disk contact graph} (WDC graph). Obviously, recognizing WDC graphs is at least as hard as the UDC graph recognition problem from Section~\ref{sec:udcg} by setting $w(v) = 1$ for every vertex $v \in V$. Accordingly, we first show that recognizing WDC graphs is \NP-hard even for stars (Section~\ref{sec:starshard}), however, embedded stars with a WDC representation can still be recognized (and one can be constructed if it exists) in linear time  (Section~\ref{sec:embeddedstars}).

\subsection{Hardness for stars}\label{sec:starshard}

We perform a polynomial reduction from the well-known 3-Partition
problem. Given a bound~$B \in \mathbb N$ and a multiset of positive
integers $\A = \{ a_1, \dots, a _{3n}\}$ such that $\frac{B}{4} < a_i
< \frac{B}{2}$ for all $i = 1, \dots, 3n$, deciding whether~$\A$ can
be partitioned into~$n$ triples of sum~$B$ each is known to be
strongly \NP-complete~\cite{Garey:1990}. Let $(\A,B)$ be a 3-Partition
instance. We construct a star $S = (V,E)$ and a radius assignment
$\mathbf{r}: V \rightarrow \mathbb{R}^+$ such that $S$ has a WDC representation respecting $\mathbf{r}$ if and only if $(\A,B)$
is a yes-instance.

\begin{wrapfigure}[11]{r}{.45\textwidth}
 \centering
  \includegraphics[trim=20mm 2mm 20mm 14mm, clip]{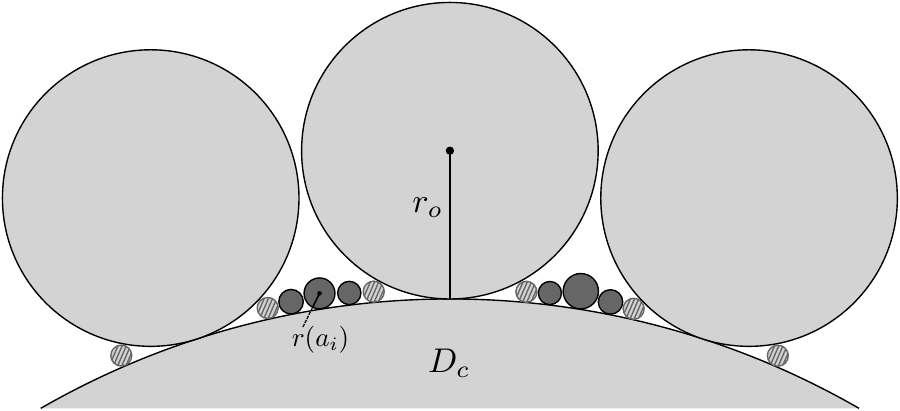}
  \caption{Reducing from 3-Partition to prove
    Theorem~\ref{thm:fr-stars-hard}. Input disks (dark) are
    distributed between gaps. Hatched disks are separators.}
  \label{fig:fr-stars-hard}
\end{wrapfigure}

We create a central disk $D_c$ of radius~$r_c$ corresponding to the
central vertex~$v_c$ of~$S$ as well as a fixed number of outer disks
with uniform radius $r_o$ chosen appropriately such that these disks
have to be placed close together around~$D_c$ without touching,
creating funnel-shaped \emph{gaps} of roughly equal size; see
Fig.~\ref{fig:fr-stars-hard}. Then, a WDC representation of $S$ exists
only if all remaining disks can be distributed among the gaps, and the
choice of the gap will induce a partition of the integers~$a_i \in
\A$. We shall represent each~$a_i$ by a single disk called an
\emph{input} disk and encode~$a_i$ in its radius. Each of the gaps is
supposed to be large enough for the input disks that represent a
\emph{feasible triple}, i.e., with sum $B$, to fit inside it, however, the gaps must be too
small to contain an \emph{infeasible triple}'s disk representation, i.e., a triple with sum $>B$.

While the principle idea of the reduction is simple, the main challenge is finding a radius assignment satisfying the above
property and taking into account numerous additional, nontrivial geometric
considerations that are required to make the construction work. For
example, we require that the lower boundary of each gap is
sufficiently flat. We achieve this by creating additional dummy gaps and
ensure that they cannot be used to realize a previously infeasible
instance.
Next, we make sure that additional
\emph{separator} disks must be placed in each gap's corners to prevent
left and right gap boundaries from interfering with the input
disks. Finally, all our constructions are required to tolerate a
certain amount of ``wiggle room'', since, firstly, the outer disks do
not touch and, secondly, some radii cannot be computed precisely in
polynomial~time. 

Since $S$ is supposed to be a star, the only adjacencies in our
construction are the ones with~$D_c$. However, several of the disks
adjacent to $D_c$ are required to be placed very close together
without actually touching. We shall, whenever we need to calculate
distances, handle these barely not touching disks as if they were
actually touching. We will describe how to compute these distances
approximately; see Lemma~\ref{lem:DTRstarNPH:chooseRad}. During this
step the radius of the central disk increases by a suitably small
amount such that no unanticipated embeddings can be created.

Let~$B > 12$ and~$n>6$, and let~$m \geq n$ be the number of gaps in
our construction. In the \emph{original} scenario described above, a
gap's boundary belonging to the central disk~$D_c$, which we call the
gap's \emph{bow}, is curved as illustrated in
Fig.~\ref{fig:theom:DTRstarNPH:originalGap}. We will, however, first
consider a \emph{simplified} scenario in which a gap is created by
placing two disks of radius~$r_o$ right next to each other on a
straight line as depicted in
Fig.~\ref{fig:theom:DTRstarNPH:SimplifiedGap}. We refer to this gap's
straight boundary as the \emph{base} of the gap. We call a point's
vertical distance from the base its \emph{height}. We also utilize the
terms \emph{left} and \emph{right} in an obvious manner. Assume for
now that we can place two \emph{separator} disks in the gap's left and
right corner, touching the base and such that the distance between the
rightmost point~$p_l$ of the left separator and the leftmost
point~$p_r$ of the right separator is exactly~$12$ units. We can
assume~$B \equiv 0\mod 4$; see Lemma~\ref{lem:DTRstarNPH:addGaps}.
Thus, we know that~$a\in \lbrace B/4+1,\dots ,B/2-1\rbrace $ for
any~$a\in \A$.

\newcommand{\lemDTRstarNPHaddGapsText}{For each~$m \geq n$, there
  exists a 3-Partition instance $(\A', B')$ equivalent to~$(\A, B)$
  with~$|\A'|=3m$ and~$B' = 180B$.}
\begin{lemma}
  \label{lem:DTRstarNPH:addGaps}
  \lemDTRstarNPHaddGapsText
\end{lemma}
\begin{proof}
  Let $n'=m-n$. For each~$a_i \in \A$, we add~$180 a_i$
  to~$\A'$. Additionally, we add~$2 n'$ integers with value~$60B-5$
  and~$n'$ integers with value~$60B+10$. The resulting
  instance~$(\mathcal A',B')$ can be realized if and only if the
  original 3-Partition instance~$(\mathcal A,B)$ is a
  yes-instance. Clearly, if~$(\mathcal A,B)$ is a yes-instance,
  then~$(\mathcal A',B')$ is a yes-instance. For the opposite
  direction, let~$S$ be a solution for~$(\mathcal A',B')$ and assume
  that~$(\mathcal A,B)$ is a no-instance. Then, there exists a
  triple~$t$ of integers in~$S$ that contains either one or two
  integers of~$\A'\setminus \{a = 180 a_i \mid a_i \in \A\}$. The sum
  of the integers in~$t$ is $5$, $10$, $20$, $50$ or $55\mod 60$
  contradicting to~$S$ being a solution for~$(\mathcal A',B')$ since
  $B'\equiv 0\mod 60$.  \qednew
\end{proof}

Our first goal is to find a function~$r:\lbrace B/4,B/4+1,\dots
,B/2\rbrace \rightarrow \mathbb{R}^+$ that assigns a disk radius to
each input integer as well as to the values~$B/4$ and~$B/2$ such that
a disk triple~$t$ together with two separator disks can be placed on
the base of a gap without intersecting each other or the outer disks
if and only if~$t$ is feasible. In the following, we show that
$r(x)=2-(4-12x/B)/B$ will satisfy our needs. We choose the radius of
the separators to be~$\rmin=r(B/4+1)=2-(1-12/B)/B$, the smallest
possible input disk radius. The largest possible input disk has
radius~$\rmax=r(B/2-1)=2+(2-12/B)/B$.  Note that~$r$ is linear and
increasing.

\begin{figure}[tb]
	\hfill
	\subcaptionbox{\label{fig:theom:DTRstarNPH:originalGap} original scenario}{\includegraphics[page=1]{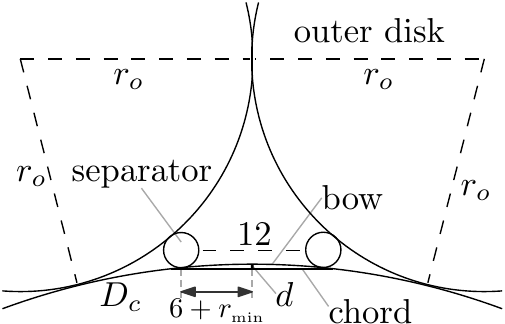}}
	\hfill
	\subcaptionbox{\label{fig:theom:DTRstarNPH:SimplifiedGap} simplified scenario}{\includegraphics[page=2]{fig/DTRstarNPGap.pdf}}
	\hfill\null
\caption{A gap, bounded in 
  \protect(\subref{fig:theom:DTRstarNPH:originalGap}) 
  by two outer disks and a bow; in 
  \protect(\subref{fig:theom:DTRstarNPH:SimplifiedGap}) the gap's base
  replaces its bow. The distance between the separators is~$12$ in
  both scenarios.}\vspace{-2ex}
\end{figure}

Next, we show for both scenarios that separators placed in each gap's
corners prevent the left and right gap boundaries from interfering
with the input disks.

\newcommand{\LemDTRstarNPHUpperBoundaryText}{For any~$a\in \A$ it is
  not possible that a disk with radius~$r(a)$ intersects one of the
  outer disks that bound the gap when placed between the two
  separators.}

\begin{lemma}
  \LemDTRstarNPHUpperBoundaryText
 \label{lem:DTRstarNPH:UpperBoundary}
\end{lemma}
\begin{proof}
  First, we utilize a geometric construction to show that~$r_o^u=38$
  is an upper bound for the outer disks' radius~$r_o$ and then use
  this result to prove that even disks with radius~$\rmax$ placed in a
  gap right next to a separator do not intersect an outer disk in the
  original scenario, implying that the input disks can actually be
  placed inside the gaps.

\begin{figure}[tbp]
  \centering
   \includegraphics[width=0.7\textwidth]{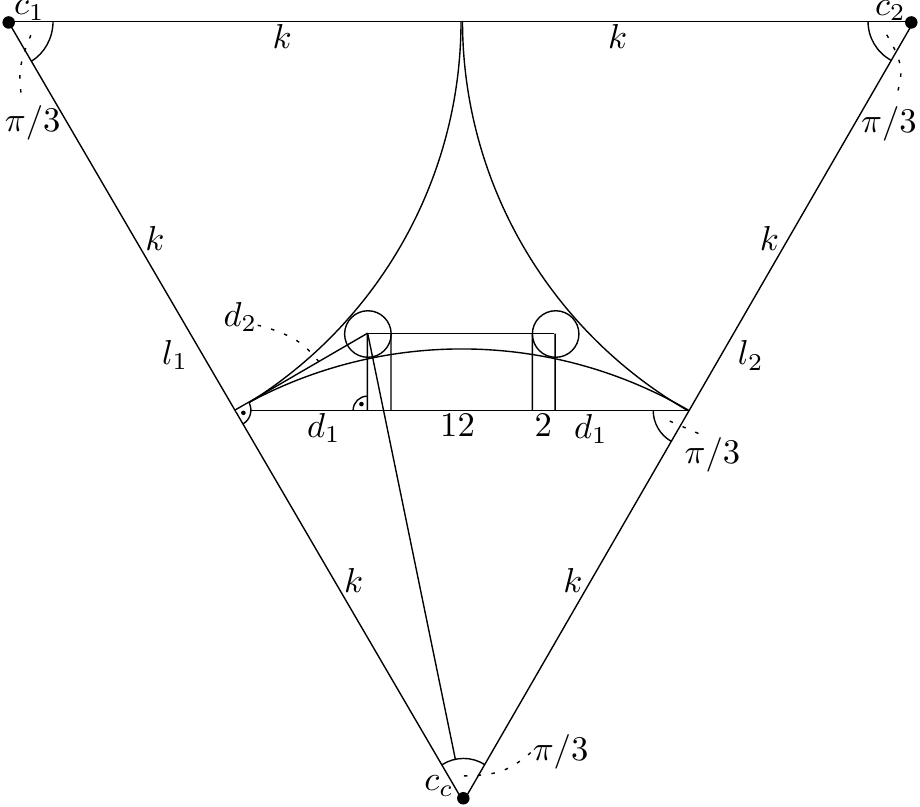}
   \caption{The geometric construction for the proof of
     Lemma~\ref{lem:DTRstarNPH:UpperBoundary}.}
  \label{fig:theom:DTRstarNPH:UpperBoundForRo}
\end{figure}

Let the distance between the two separator disks always remain~12 as
in Fig.~\ref{fig:theom:DTRstarNPH:SimplifiedGap}. For fixed~$\rmin$,
if the number of gaps~$m$ decreases, then~$r_c$ decreases and~$r_o$
increases. For fixed~$m$, if~$\rmin$ increases, so does~$r_o$.
We designate a minimum value of~$m_\textnormal{min}=6$ to~$m$ and
observe that~$\rmin=r(B/4+1)=2-(1-12/B)/B<2$ for any~$B>12$. Consider
the extreme case~$m=6$ and~$\rmin=2$ in
Fig.~\ref{fig:theom:DTRstarNPH:UpperBoundForRo}. The angle between the
two line segments~$l_1,l_2$ bounded by the centers~$c_1,c_2$ of two
adjacent outer disks and the center~$c_c$ of the central disk~$D_c$
is~$2\pi /m_\textnormal{min}=\pi/3$ and, therefore, $l_1$ and~$l_2$
together with the line segment~$\overline{c_1c_2}$ constitute an
equilateral triangle. This implies that the outer disk radius is equal
to the radius of the central disk, we denote this radius with~$k$. By
using basic trigonometry as well as the Pythagorean Theorem, we obtain
that~$k$ has to satisfy the equality~$k=d_1+2+12+2+d_1=2d_1+16$,
where~$d_1=\mathrm{cos}(\pi/6)\cdot d_2=\sqrt{3}\cdot d_2/2$
and~$d_2=\sqrt{(k+2)^2-k^2}=\sqrt{4k+4}=2\sqrt{k+1}$. This set of
equalities solves for~$k=22+4\sqrt{5}\sqrt{3}=37.4919...<38=r_o^u$.

\begin{figure}[htb]
  \centering
   \includegraphics{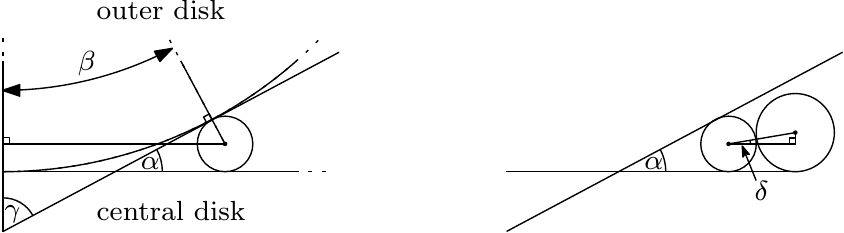}
   \caption{A disk with radius~$\rmax$ can always be placed right next
     to a separator in a corner of a gap.}
  \label{fig:theom:DTRstarNPH:rMax}
\end{figure}

We now show that a disk with radius~$\rmax$ placed in a gap right next
to a separator does not intersect an outer disk. Assume the separator
always touches both the outer and the central disk, and consider the
angle~$\alpha$ between the two tangents on the separator in the two
touching points. For fixed~$r_o$ and~$r_c$, if~$\rmin$ decreases, so
does~$\alpha$. For fixed~$\rmin$ and~$r_o$, if~$r_c$ increases,
then~$\alpha$ decreases. Similarly, for fixed~$\rmin$ and~$r_c$,
if~$r_o$ increases, then~$\alpha$ decreases. Thus, to compute the
lower bound for~$\alpha$, we use upper bounds~$r_c^u$ and~$r_o^u$
for~$r_c$ and~$r_o$ respectively and a lower bound~$\rmin^l$
for~$\rmin$. We will show that a disk with radius~$\rmax$ always fits
inside~$\alpha$ right next to the separator. A suitable choice for
these values is $r_c^u = \infty$, $r_o^u = 38$ and~$\rmin^l = 2 -
\frac{1}{12}$.

To compute the value of~$\alpha$ corresponding to these values,
consider Fig.~\ref{fig:theom:DTRstarNPH:rMax}. It holds: $\beta =
\arccos((r_o^u-\rmin^l)/(r_o^u+\rmin^l))$, $\gamma = \pi/2 - \beta$
and~$\alpha =\pi/2 - \gamma \approx 25.3^\circ$. Recall that~$\rmax^u
= 2 + 1/6$ is an upper bound for~$\rmax$. Assume disks with
radii~$\rmin$ and~$\rmax$ are placed next to each other on a
horizontal line. Then, for the angle~$\delta$ in
Fig.~\ref{fig:theom:DTRstarNPH:rMax} it holds: $\delta =
\arcsin((\rmax - \rmin)/(\rmax + \rmin)) \leq \arcsin((\rmax^u -
\rmin^l)/(\rmax^u + \rmin^l)) \approx 3.51^\circ < \alpha/2$. It
follows that the center of the bigger disk lies below the bisector
of~$\alpha$. Therefore, the bigger disk fits
inside~$\alpha$. \qednew
\end{proof}

For our further construction, we need to prove the following property.
\newcommand{\PrDTRstarNPHfitIffFeasibleText}{Each feasible triple fits
  inside a gap containing two separators and no infeasible triple
  does.}

\begin{property}
  \PrDTRstarNPHfitIffFeasibleText
  \label{pr:DTRstarNPH:fitIffFeasible}
\end{property}
It can be easily verified that for~$x_1$,$x_2$,$x_3$, $\sum_{i=1}^3
x_i \leq B$, it is $2 \sum_{i=1}^3 r(x_i) \leq 12$, implying the first
part of Property~\ref{pr:DTRstarNPH:fitIffFeasible}. We define
$s_i=2\rmin+2\sqrt{(\rmax+\rmin)^2-(\rmax-\rmin)^2}$. In the proof of
Lemma~\ref{lem:DTRstarNPH:simplified:triples}, we will see that~$s_i$
is the horizontal space required for the triple~$(\rmin,\rmax,\rmin)$,
which is the narrowest infeasible triple.
Next, let
$d(\varepsilon,x)=\sqrt{(r(x)-\varepsilon/2)^2+(r(x)-\rmin)^2}$
for~$\eps > 0$ and $x \in \{ B/4+1,\dots
,B/2-1\}$. %
We will see that~$d(\varepsilon,x)$ is an upper bound for the distance
between the center of a disk~$D(x)$ with radius~$r(x)$ and the
rightmost (leftmost) point of the left (right) separator disk, if the
overlap of their horizontal projections is at least~$\eps/2$%
.
\newcommand{\lemDTRstarNPHSimplifiedTriplesText}{There exist $\eps >0$
  and~$\eps_1, \eps_2, \phi \ge 0$ with~$\eps=\eps_1+\eps_2$ which
  satisfy the two
  conditions:\newcond{eq:DTRstarNPH:Cond01}~$12+\varepsilon \le s_i$
  and\newcond{eq:DTRstarNPH:Cond02}~$d(\varepsilon_1,x)\le r(x)-\phi
  \forall x\in \lbrace B/4+1,\dots ,B/2-1\rbrace$.
  These conditions imply the second part of
  Property~\ref{pr:DTRstarNPH:fitIffFeasible} for the simplified
  scenario.}
\begin{lemma}\label{lem:DTRstarNPH:simplified:triples}
\lemDTRstarNPHSimplifiedTriplesText
\end{lemma}
\begin{proof}
  Recall that the function~$r$ is linear. A triple of disks with
  uniform radius~$r(B/3)=2$ requires a total horizontal space
  of~$2\cdot 2\cdot 3=12$ if placed tightly next to each other on a
  straight line. Therefore, since the radius~$\rmin$ of the
  separators is less than~2, it follows that every feasible disk
  triple fits in the gap since (1) a triple of disks with uniform
  radius~$r(B/3)$ yields an upper bound for the amount of horizontal
  space required by any feasible disk triple and since (2)
  $r(B/3)>\rmin$, which implies that if the three disks are placed
  next to each other on the base, the height of the leftmost point
  of the disk triple is greater than the height of the rightmost
  point~$p_l$ of the left separator and, therefore, these disks do not
  touch (and the same holds true for the right side respectively), see
  Fig.~\ref{fig:theom:DTRstarNPH:maximumFeasibleOnBase}.

  \begin{figure}[tb]
  \centering
	  \subcaptionbox{\label{fig:theom:DTRstarNPH:maximumFeasibleOnBase}}{\includegraphics[page=1,width=0.8\textwidth]{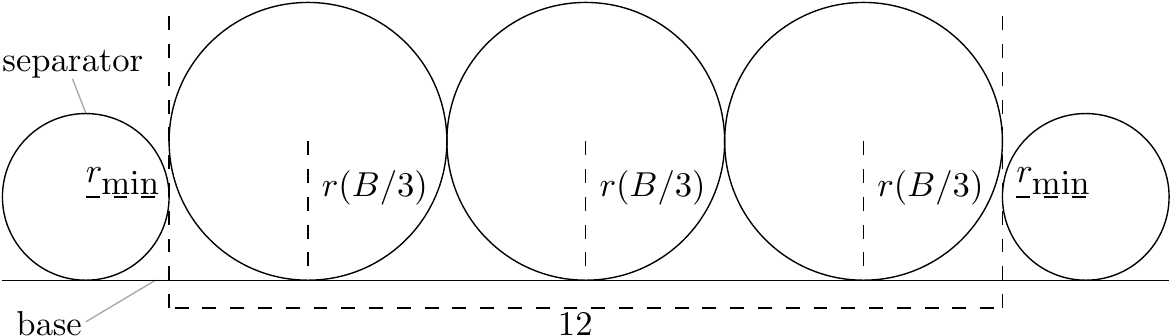}}
	  \hfill
	  \subcaptionbox{\label{fig:theom:DTRstarNPH:smallerSimplifiedGap}}{\includegraphics[page=1,width=0.8\textwidth]{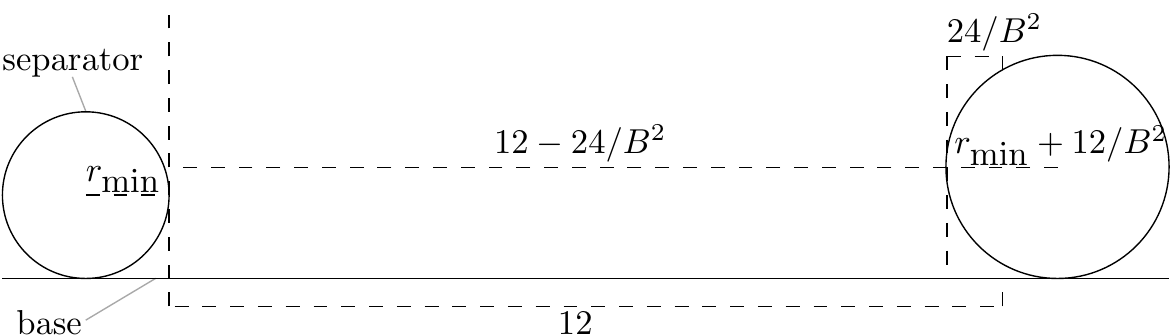}}

      \caption{Two illustrations regarding the proof of
        Lemma~\ref{lem:DTRstarNPH:simplified:triples}. \protect\subref{fig:theom:DTRstarNPH:maximumFeasibleOnBase}
        Depiction of a feasible input triple's disk
        representation. This particular representation requires the
        largest possible amount of horizontal space out of all
        representations for feasible input triples with sum~$B$.
        \protect\subref{fig:theom:DTRstarNPH:smallerSimplifiedGap}An
        upper bound for the amount of horizontal space between the two
        disks placed in a gap's corner.}
\end{figure}

Next, consider the disk triple~$t_i=(\rmin,\rmax,\rmin)$. The
sum of the integers corresponding to the disks of~$t_i$
is~$B/4+1+B/2-1+B/4+1=B+1$ and, therefore,~$t_i$ is infeasible. When
placing the three disks next to each other on a straight line, placing
the disk with radius~$\rmax$ in the middle maximizes the difference
between the radii of adjacent disks and, therefore, minimizes the
overall horizontal space required, which, using the Pythagorean
Theorem, can be described as
$s_i=2\rmin+2\sqrt{(\rmax+\rmin)^2-(\rmax-\rmin)^2}$, see
Fig.~\ref{fig:theom:DTRstarNPH:minimumInfeasibleOnBase}. Since~$r$
is linear and $B+1$ is the smallest possible sum of any infeasible
integer triple,~$s_i$ is the least possible amount of horizontal space
required by any infeasible disk triple.

\begin{figure}[tb] 
 \hfill \subcaptionbox{\label{fig:theom:DTRstarNPH:minimumInfeasibleOnBase}}{\includegraphics[page=1]{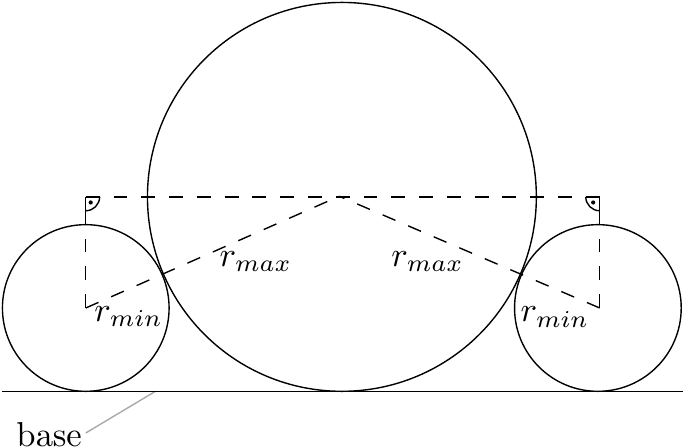}}
 \hfill \subcaptionbox{\label{fig:theom:DTRstarNPH:dxPl}}{\includegraphics[page=2]{fig/DTRstarNPHminimumInfeasibleOnBottom.pdf}}
 \hfill\null
  \caption{By Lemma~\ref{lem:DTRstarNPH:simplified:triples}, it is not
    possible to place an infeasible disk triple inside a simplified
    gap. \protect\subref{fig:theom:DTRstarNPH:minimumInfeasibleOnBase}~The
    smallest possible infeasible disk triple.
    \protect\subref{fig:theom:DTRstarNPH:dxPl}~Disk~$D(x)$ is
    intersecting the separator.}
\end{figure}

In order to show that $0<\varepsilon \le 17/B^2$ is a sufficient
choice to satisfy Condition~\ref{eq:DTRstarNPH:Cond01}, for an
arbitrary~$0<c\le 17$ we assign~$\varepsilon =c/B^2$ and show that the
condition holds and for any $B>12$.

\begin{align*}
  12+\varepsilon \le s_i\Leftrightarrow \\
  12+c/B^2\le 2\rmin+2\sqrt{(\rmax+\rmin)^2-(\rmax-\rmin)^2}\Leftrightarrow \\
  3+(c/4)/B^2-(1/2)\rmin\le \sqrt{\rmax\rmin}\Leftarrow \\
  (3+(c/4)/B^2-(1/2)(2-1/B+12/B^2))^2\le~~~~\\
  (2-1/B+12/B^2)(2+2/B-12/B^2)\Leftrightarrow \\
  9+(c^2/16)/B^4+1+1/(4B^2)+36/B^4-1/B+12/B^2-6/B^3+~~~~\\
  3c/(2B^2)-6~+3/B-36/B^2-c/(2B^2)+c/(4B^3)-3c/B^4\le~~~~\\
  4+2/B-2/B^2+36/B^3-144/B^4\Leftrightarrow \\
  (c^2/16-3c+180)/B^4+(c/4-42)/B^3+(c-87/4)/B^2\le 0\Leftarrow \\
  c^2/16-3c+180+(c/4-42)B+(c-87/4)B^2\le 0\Leftarrow \\
  17/16-3\cdot 0+180+(17/4-42)B+(17-87/4)B^2\le 0\Leftrightarrow \\
  17/16+180-(151/4)B-(19/4)B^2\le 0~~~~
\end{align*}

The last inequality clearly holds true for any~$B>12$, which concludes
the proof of this step.

In order to show that $\varepsilon_1=16/B^2$ and $0\le \phi \le 1/B^2$
satisfy Condition~\ref{eq:DTRstarNPH:Cond02}, we substitute~$y=x\cdot
12/B$ and show that
\begin{align*} d((16/B^2),(y\cdot B/12))\le
  r(y\cdot B/12)-c/B^2
\end{align*}
for any~$y\in \lbrace 3+1\cdot 12/B, 3+2\cdot 12/B,\dots
,6-12/B\rbrace $, any~$0\le c\le 1$ and any~$B>12$.
\begin{align*}
  d((16/B^2),(y\cdot B/12))\le r(y\cdot B/12)-c/B^2 \Leftrightarrow \\
  \sqrt{(r(y\cdot B/12)-8/B^2)^2+(r(y\cdot B/12)-\rmin)^2}\le r(y\cdot B/12)-c/B^2\Leftarrow \\
  2\cdot r(y\cdot B/12)^2-16\cdot r(y\cdot B/12)/B^2+64/B^4-2r(y\cdot B/12)\rmin+(\rmin)^2\le ~~~~\\
  r(y\cdot B/12)^2-2c\cdot r(y\cdot B/12)/B^2+c^2/B^4\Leftrightarrow \\
  r(y\cdot B/12)^2-16\cdot r(y\cdot B/12)/B^2+64/B^4-2r(y\cdot B/12)\rmin+(\rmin)^2\le ~~~~\\
  -2c\cdot r(y\cdot B/12)/B^2+c^2/B^4\Leftrightarrow \\
  (4+16/B^2+y^2/B^2-16/B+4y/B-8y/B^2)+(-32/B^2+64/B^3-16y/B^3)+~~~~\\
  64/B^4+(-8+16/B-4y/B+4/B-8/B^2+2y/B^2-48/B^2+~~~~\\
  96/B^3-24y/B^3)+(4+1/B^2+144/B^4-4/B+48/B^2-24/B^3)+~~~~\\
  (4c/B^2 -8c/B^3+2cy/B^3)-c^2/B^4\le 0\Leftrightarrow \\
  (208-c^2)/B^4+(64-16y+96-24y-24-8c+2cy)/B^3+~~~~\\
  (16+y^2-8y-32-8+2y-48+1+48+4c)/B^2\le 0\Leftrightarrow \\
  (208-c^2)/B^4+(136-40y+2cy-8c)/B^3+(-23+y^2-6y+4c)/B^2\le 0\Leftarrow \\
  208-c^2+(136-40y+2cy-8c)B+(-23+y^2-6y+4c)B^2\le 0\Leftarrow \\
  208-0^2+(136-40\cdot 3+2\cdot 1\cdot 6-8\cdot 0)B+(-23+0+4\cdot 1)B^2\le 0\Leftrightarrow \\
  208+28B-19B^2\le 0~~~~
\end{align*}

The last inequality clearly holds true for any $B>12$, which concludes
this part of the proof.

Finally, we show that no infeasible disk triple triple can be placed
in the gap together with two separators.  Recall that the distance
between the rightmost point~$p_l$ of the left separator and the
leftmost point~$p_r$ of the right separator, which are located at
height~$\rmin$, is exactly~$12$
units. Condition~\ref{eq:DTRstarNPH:Cond01} ensures that all
infeasible disk triples take up at least~$12+\varepsilon $ units of
horizontal space, however, this condition is not sufficient to
guarantee that infeasible disk triples can not be placed between the
separators since we do not know yet at what height the leftmost and
the rightmost point of the disk triple are located. However, it is
guaranteed that either the leftmost point of the disk triple is
located at least~$\varepsilon_1 /2$ units to the left of~$p_l$ or the
rightmost point of the triple is located at least~$\varepsilon_1 /2$
units to the right of~$p_r$. Let now~$x\in \lbrace B/4+1,\dots
,B/2-1\rbrace$ be an input integer. The Pythagorean Theorem implies
that the distance between~$p_l$~($p_r$) and the center of a
disk~$D(x)$ with radius~$r(x)$ whose center is located between the two
separator disks and whose leftmost (rightmost) point is located at
least~$\varepsilon_1 /2$ units to the left (right) of~$p_l$~($p_r$) is
at
most~$d(\varepsilon_1,x)=\sqrt{(r(x)-\varepsilon_1/2)^2+(r(x)-\rmin)^2}$,
as illustrated in
Fig.~\ref{fig:theom:DTRstarNPH:dxPl}. Condition~\ref{eq:DTRstarNPH:Cond02}
ensures that this distance is at most $r(x)-\phi$, implying
that~$D(x)$ intersects the left (right) separator. Therefore,
Condition~\ref{eq:DTRstarNPH:Cond01} and
Condition~\ref{eq:DTRstarNPH:Cond02} together guarantee that
infeasible disk triples together with two separators can not be placed
inside a gap in the simplified scenario. The significance
of~$\varepsilon_2$ becomes clear in the proof of
Lemma~\ref{lem:DTRstarNPH:nonTwelve}, where we tailor our
conditions to apply to the original scenario as well. \qednew
\end{proof}

So far we assumed that the separators are always placed in the corners
of the gap. But in fact, separators could be placed in a different
location, moreover, there could even be gaps with multiple separators
and gaps with zero or one separator. Since the radius of the
separators is~$\rmin$, which is the radius of the smallest possible
input disk, it seems natural to place them in the gaps' corners to
efficiently utilize the horizontal space. However, all feasible disk
triples (except~$(B/3,B/3,B/3)$) require less than~$12$ units of
horizontal space. It might therefore be possible to place a feasible
disk triple inside a gap together with two disks that are not
necessarily separators but input disks with a radius greater
than~$\rmin$. 
To account for this problem, we prove the following property.

\newcommand{\prDTRstarNPHfeasibleOnlyWithSepText}{A feasible disk
  triple can be placed in the gap together with two other disks only
  if those two disks are separators.}

\begin{property}
  \prDTRstarNPHfeasibleOnlyWithSepText
  \label{pr:DTRstarNPH:feasibleOnlyWithSep}
\end{property}

We define
$s_f=2r(B/4)+2\sqrt{(r(B/2)+r(B/4))^2-(r(B/2)-r(B/4))^2}$. In the
proof of Lemma~\ref{lem:DTRstarNPH:simplified:separators}, we will see
that~$s_f$ is a lower bound for the horizontal space consumption of
any feasible triple.

\newcommand{\lemDTRstarNPHSimplifiedSeparatorsText}{%
  There exist $\xi >0$ and~$\xi_1, \xi_2, \psi \ge 0$
  with~$\xi=\xi_1+\xi_2$ satisfying the following two
  conditions: \newcond{eq:DTRstarNPH:Cond03} $12-24/B^2+\xi \le s_f$
  and\newcond{eq:DTRstarNPH:Cond04} $d(\xi_1,x)\le r(x)-\psi\, \forall
  x\in \lbrace B/4+1,\dots ,B/2-1\rbrace$.
  These conditions imply
  Property~\ref{pr:DTRstarNPH:feasibleOnlyWithSep} for the simplified
  scenario.}
\begin{lemma}
\label{lem:DTRstarNPH:simplified:separators}
\lemDTRstarNPHSimplifiedSeparatorsText
\end{lemma}
\begin{proof}

  In order to show that $0<\xi \le 17/B^2$ is a sufficient choice to
  satisfy Condition~\ref{eq:DTRstarNPH:Cond03}, for an
  arbitrary~$0<c\le 17$ we assign~$\xi =c/B^2$ and show that the
  condition holds for any $B>12$.
\begin{align*}
    12-24/B^2+\xi \le s_f\Leftrightarrow \\
    12-24/B^2+c/B^2\le 2r(B/4)+2\sqrt{(r(B/2)+r(B/4))^2-(r(B/2)-r(B/4))^2}\Leftrightarrow \\
    12+(c-24)/B^2-2r(B/4)\le 4\sqrt{r(B/2)\cdot r(B/4)}\Leftrightarrow \\
    2+(c/4-6)/B^2+(1/2)/B\le \sqrt{(2+2/B)\cdot (2-1/B)}\Leftarrow \\
    4+(c^2/16-3c+36)/B^4+(1/4)/B^2+(c-24)/B^2+2/B+(c/4-6)/B^3\le ~~~~\\
    4-2/B+4/B-2/B^2\Leftrightarrow \\
    (c^2/16-3c+36)/B^4+(c/4-6)/B^3+(c-24+1/4+2)/B^2\le 0\Leftarrow \\
    c^2/16-3c+36+(c/4-6)B+(c-22+1/4)B^2\le 0\Leftarrow \\
    17^2/16-3\cdot 0+36+(17/4-6)B+(17-22+1/4)B^2\Leftrightarrow \\
    289/16+36-(7/4)B-(19/4)B^2\le 0~~~~
\end{align*}

The last inequality can easily be verified to be true for any~$B>12$. %

To prove Condition~\ref{eq:DTRstarNPH:Cond04}, we choose
$\xi_1=16/B^2$ and an arbitrary~$ \psi$ with $0\le \psi \le
1/B^2$. The remaining proof of Condition~\ref{eq:DTRstarNPH:Cond04} is
identical to that of Condition~\ref{eq:DTRstarNPH:Cond02}.

Recall that the second smallest possible input disk radius
is~$r(B/4+2)=2-(1-24/B)/B=2-(1-12/B)/B+12/B^2=\rmin+12/B^2$ and,
therefore,~$12-2\cdot 12/B^2=12-24/B^2$ is an upper bound for the
remaining horizontal space in a gap in which two disks have been
placed such that one of the disks has radius greater than~$\rmin$, see
Fig.~\ref{fig:theom:DTRstarNPH:smallerSimplifiedGap}. The input
integers' values are at least~$B/4+1$ and at most~$B/2-1$, therefore,
the horizontal space consumption of the disk
triple~$t_f=(r(B/4),r(B/2),r(B/4))$ is a lower bound for the space
consumption of any feasible disk triple since the total difference
between the radii of adjacent disks in~$t_f$ is larger than that of
any feasible disk triple. Yet again we utilize the Pythagorean Theorem
to describe~$t_f$'s required horizontal space
as~$s_f=2r(B/4)+2\sqrt{(r(B/2)+r(B/4))^2-(r(B/2)-r(B/4))^2}$.
Condition~\ref{eq:DTRstarNPH:Cond03} therefore ensures that any
feasible disk triple consumes at least~$12-24/B^2+\xi$ horizontal space and, analogously to
Condition~\ref{eq:DTRstarNPH:Cond02},
Condition~\ref{eq:DTRstarNPH:Cond04} together with
\begin{equation*}
d(\xi_1,x)=\sqrt{(r(x)-\xi_1/2)^2+(r(x)-\rmin)^2}  
\end{equation*}
ensures that one of the disks of~$t_f$ intersects a separator or one
of the replacing disks, implying
Property~\ref{pr:DTRstarNPH:feasibleOnlyWithSep}.  Like
with~$\varepsilon_2$ the significance of~$\xi_2$ will become apparent
in the proof of Lemma~\ref{lem:DTRstarNPH:nonTwelve} when we describe
how to apply our conditions to the original scenario.
\hfill\qednew
\end{proof}
We verify in the proofs of Lemmas~\ref{lem:DTRstarNPH:simplified:triples}
and~\ref{lem:DTRstarNPH:simplified:separators} that choosing~$\varepsilon_1,\xi_1=16/B^2$
and~$\varepsilon_2,\phi ,\xi_2 ,\psi = 1/B^2$ satisfies our four
conditions. %

Intuitively, Conditions~(\ref{eq:DTRstarNPH:Cond01})--(\ref{eq:DTRstarNPH:Cond04}) have the following
meaning. By~(\ref{eq:DTRstarNPH:Cond01}), the horizontal space
consumption of any infeasible triple is greater than~12 by some fixed
buffer. By~(\ref{eq:DTRstarNPH:Cond03}), the horizontal space
consumption of any feasible triple is very close
to~12. Conditions~(\ref{eq:DTRstarNPH:Cond02})
and~(\ref{eq:DTRstarNPH:Cond04}) imply that if the overlap of the
horizontal projections of a separator and an input disk is large enough,
the two disks intersect, implying that triples with sufficiently large
space consumption can indeed not be placed between two separators.

In the original scenario, consider a straight line directly below the
two separators. We call this straight line the gap's \emph{chord}, see
Fig.~\ref{fig:theom:DTRstarNPH:originalGap}. The gap's chord has a
function similar to the base in the simplified scenario. We still want
separators to be placed in the gap's corners. The distance between the
rightmost point~$p_l$ of the left separator and the leftmost
point~$p_r$ of the right separator is now allowed to be slightly
more than~$12$.
The horizontal space consumption of a disk triple placed on
the bow is lower compared to the disk triple being placed on the
chord. Moreover, the overlap of the horizontal projections of a
separator and an input disk can now be bigger without causing an
intersection. However, we show that if the maximum distance~$d$
between a gap's bow and its chord is small enough, the original
scenario is sufficiently close to the simplified one, and the four
conditions still hold, implying the desired properties.

\newcommand{\lemDTRstarNPHnonTwelveText}{In the original scenario,
  let~$d \leq 1/4B^2$, and let the amount of free horizontal space in
  each gap after inserting the two separators in each corner be
  between~12 and~$12 + 1 / 4B^2$. Then,
  Properties~\ref{pr:DTRstarNPH:fitIffFeasible}
  and~\ref{pr:DTRstarNPH:feasibleOnlyWithSep} still hold.}
\begin{lemma}
  \label{lem:DTRstarNPH:nonTwelve}
  \lemDTRstarNPHnonTwelveText
\end{lemma}
\begin{proof}
  Obviously, each feasible triple can still fit in the gap together
  with two separators.

\begin{figure}[tb]
	\hfill
	\subcaptionbox{\label{fig:theom:DTRstarNPH:maxHorizontalSave01}}{\includegraphics[page=1]{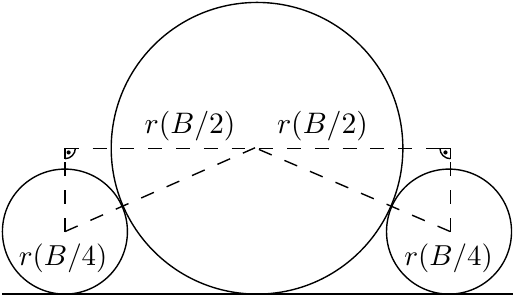}}
	\hfill
	\subcaptionbox{\label{fig:theom:DTRstarNPH:maxHorizontalSave02}}{\includegraphics[page=2]{fig/DTRstarNPHmaxHorizontalSave.pdf}}
	\hfill\null
  \caption{An upper bound for the amount of horizontal space that can
    be saved by placing a disk triple on the gap's bow instead of its
    chord can be calculated by replacing the bow by a straight
    line~$d$ units above the chord and comparing the required space to
    the space required when placing the two outer disks directly on
    the chord.}
  \label{fig:theom:DTRstarNPH:maxHorizontalSave}
\end{figure}

  We now compute the amount of horizontal space that can be saved in
  the original scenario compared to the simplified scenario when
  placing any infeasible or feasible disk triple on the bow instead of
  on the chord.

  Consider the disk triple~$t=(r(B/4),r(B/2),r(B/4))$. Yet again
  utilizing the Pythagorean Theorem, we calculate an upper bound for
  the amount of horizontal space that can be saved to be
  \begin{align*}
    s&=2(\sqrt{(r(B/2)+r(B/4))^2-(r(B/2)-r(B/4))^2}\\
    &-\sqrt{(r(B/2)+r(B/4))^2-(r(B/2)+d-r(B/4))^2})
  \end{align*}
  by simply moving the left and right disk down by~$d$ units and as
  far to the center disk as possible, see
  Fig.~\ref{fig:theom:DTRstarNPH:maxHorizontalSave}. Recall
  that~$B/4$ is smaller and~$B/2$ is greater than any input integer,
  and the differences between the radii of adjacent disks in~$t$ are
  larger than in any actual input disk triple. Therefore, the
  value~$s$ is also an upper bound for the amount of horizontal space
  that can be saved in the original scenario compared to the
  simplified scenario when placing any infeasible or feasible disk
  triple.

  We now show that for~$d\le 1/4B^2$ an upper bound for~$s$
  is~$1/4B^2$.

\begin{align*}
s\le 1/4B^2\Leftrightarrow \\
2(\sqrt{(r(B/2)+r(B/4))^2-(r(B/2)-r(B/4))^2}~~~~\\
-\sqrt{(r(B/2)+r(B/4))^2-(r(B/2)+d-r(B/4))^2})\le 1/4B^2\Leftrightarrow \\
2(\sqrt{4\cdot r(B/2)\cdot r(B/4)}~~~~\\
-\sqrt{4\cdot r(B/2)\cdot r(B/4)-2d\cdot r(B/2)+2d\cdot r(B/4)-d^2})\le 1/4B^2\Leftrightarrow \\
\sqrt{4\cdot r(B/2)\cdot r(B/4)}-1/8B^2~~~~\\
\le \sqrt{4\cdot r(B/2)\cdot r(B/4)-2d\cdot r(B/2)+2d\cdot r(B/4)-d^2}\Leftarrow \\
1/64B^4-\sqrt{4\cdot r(B/2)\cdot r(B/4)}/4B^2\le 2d\cdot r(B/4)-2d\cdot r(B/2)-d^2\Leftrightarrow \\
1/16B^2+8dB^2\cdot (2+2/B)-8dB^2\cdot (2-1/B)+4d^2B^2~~~~\\
\le \sqrt{4\cdot (2+2/B)\cdot (2-1/B)}\Leftarrow \\
(1/32B^2+4dB^2\cdot (2+2/B-2+1/B)+2d^2B^2)^2\le (2+2/B)\cdot (2-1/B)\Leftarrow \\
(5/32B^2+1\cdot (3/B))^2\le 4-2/B+4/B-2/B^2\Leftrightarrow \\
25/1024B^4+30/32B^3+9/B^2\le 4-2/B+4/B-2/B^2\Leftrightarrow \\
25/1024B^4+15/16B^3+11/B^2-2/B-4\le 0\Leftarrow \\
25/1024+15B/16+11B^2-2B^3-4B^4\le 0~~~~
\end{align*}
The last inequality clearly holds true for any $B > 12$.

Consider an infeasible triple placed in a gap with two separators in
the corners. The triple requires horizontal space at least $s_i - s
\geq 12 + \eps_1 + \eps_2 - s \geq 12 + (16 + \frac{3}{4})/B^2$, and
at most $12 + 1/4B^2$ is available. Thus, \wlg, the left separator and
the leftmost disk~$D(x)$ of the triple have overlap of horizontal
projections of at least $(16 + 1/2B^2)/2 = 8 + 1/4B^2$. If~$D(x)$
would be placed on the chord instead of the bow, the distance
between~$p_l$ and the center of~$D(x)$ would be at most~$d(\eps_1, x)
\leq r(x) - \phi = r(x) - 1/B^2$. When~$D(x)$ is moved up and placed
on the bow instead, this distance remains at most~$r(x) - 1/B^2 + d
\leq r(x) - 3/4B^2$. Thus, the infeasible triple doesn't fit inside
the gap, and Property~\ref{pr:DTRstarNPH:fitIffFeasible} holds.

Now consider a feasible triple and assume that a separator has been
replaced by a bigger disk in the gap's corner. The triple requires
horizontal space at least $s_f - s \geq 12 - 24/B^2 + \xi_1 + \xi_2 -
1/4B^2 = 12 - 7/B^2 - 1/4B^2$. Replacing a separator by a bigger disk
in the gap's corner consumes at least $24/B^2$ horizontal space in the
simplified scenario (see
Fig.~\ref{fig:theom:DTRstarNPH:smallerSimplifiedGap}) and even more in
the original scenario. Then, \wlg, the overlap of the horizontal
projections of~$D(x)$ and the disk in the left gap corner is at
least~$((12 - 7/B^2 - 1/4B^2) - (12 - 24/B^2 + 1/4B^2))/2 = 8/B^2 +
1/4B^2$, and, analogously to the above argument, the two disks
intersect. Therefore, the triple can not fit in the gap, and
Property~\ref{pr:DTRstarNPH:feasibleOnlyWithSep} follows. \qednew
\end{proof}

In order to conclude the hardness proof, it therefore remains to describe how
to choose the radii for the central and outer disks and how to create
the gaps such that~$d\le 1/4B^2$.

Recall that we have a central disk~$D_c$ with radius~$r_c$ and~$m$
outer disks with radius~$r_o$ which are tightly packed around~$D_c$
such that~$m$ equal-sized gaps are created.  With basic trigonometry
we see that~$r_c+r_o=r_o/\sin(\pi /m)$ %
and, therefore,~$r_c=r_o/\sin(\pi /m)-r_o$. Clearly, there always
exists a value~$r_o$ such that the two separator disks can be placed
in each gap's corners and such that the distance between each pair of
separators is exactly~$12$ units. Let~$\roexact$ be this value.
Moreover, the maximum distance~$d$ between a gap's bow and its chord
is of particular importance, see
Fig.~\ref{fig:theom:DTRstarNPH:originalGap}. Using the Pythagorean
Theorem, it can be calculated to
be~$d=r_c-(\sqrt{(r_c+\rmin)^2-(6+\rmin)^2}-\rmin)$. The crucial
observation is that we do not necessarily need to
choose~$m=n$. Instead we may choose any~$m\ge n$ and thereby
decrease~$d$, as long as we make sure that~$m$ is still a polynomial
in the size of the input or numeric values and that the~$m-n$
additional gaps cannot be used to solve an instance which should be
infeasible.

\newcommand{\lemDTRstarNPHchooseRadText}{There exist constants~$c_1$,
  $c_3$, $c_4$, such that for~$m=B^{c_1}$, $\eps_3=1/B^{c_3}$
  and~$\eps_4=1/B^{c_4}$, there exist values $\roapprox$ for~$r_o$
  and~$\rcapprox$ for $r_c$, for which it holds $\roexact < \roapprox
  \le \roexact +\eps_3$ and~$\rcexact < \rcapprox \le \rcexact+\eps_4$
  for $\rcexact = \roapprox / \sin(\pi /m)-\roapprox$.
  Moreover, the constants can be chosen such that~$d \leq 1/4B^2$ and
  such that the amount of free horizontal space in each gap is
  between~$12$ and~$12 + 1/4B^2$.
  Finally, $\roapprox$ and~$\rcapprox$ can be computed in polynomial
  time.}
\begin{lemma}
  \label{lem:DTRstarNPH:chooseRad}
  \lemDTRstarNPHchooseRadText
\end{lemma}
\begin{proof}
  \textbf{Choosing~$m$}. In order to choose an~$m\ge n$ such
  that~$d\le 1/4B^2$, we require some information about the
  radius~$r_o$ of the outer disks. A precise calculation of this value
  yields a complicated formula, however, a lower as well as an upper
  bound for~$r_o$ are sufficient to conclude our
  argument. Clearly,~$r_o^l=6$ is a lower bound for~$r_o$. In the
  proof of Lemma~\ref{lem:DTRstarNPH:UpperBoundary}, we have shown
  that for $m \geq m_{min}=6$, $r_o^u=38$ is an upper bound
  for~$r_o$. Recalling that~$\rmin$ is a polynomial in~$B$,
  that~$m\ge m_{min}=6$ and utilizing that~$\sin(x)\le x,\forall x\ge
  0$, we can now prove that~$m$ can be chosen as a polynomial in~$B$
  such that~$d\le 1/4B^2$:

\begin{align*}
d\le 1/4B^2\Leftrightarrow \\
r_c-(\sqrt{(r_c+\rmin)^2-(6+\rmin)^2}-\rmin)\le 1/4B^2\Leftrightarrow \\
(r_c+\rmin)-1/4B^2\le \sqrt{(r_c+\rmin)^2-(6+\rmin)^2}\Leftarrow \\
1/16B^4-(r_c+\rmin)/2B^2\le -(6+\rmin)^2\Leftrightarrow \\
1/8B^2+2B^2(6+\rmin)^2-\rmin\le r_c\Leftrightarrow \\
1/8B^2+2B^2(6+\rmin)^2-\rmin+r_o\le r_o/\mathrm{sin}(\pi /m)\Leftrightarrow \\
\mathrm{sin}(\pi /m)\le r_o/(1/8B^2+2B^2(6+\rmin)^2-\rmin+r_o)\Leftarrow \\
\end{align*}
\begin{align*}
\pi /m\le r_o/(1/8B^2+2B^2(6+\rmin)^2-\rmin+r_o)\Leftrightarrow \\
m\ge (\pi /r_o)\cdot(1/8B^2+2B^2(6+\rmin)^2-\rmin+r_o)\Leftarrow \\
m\ge (\pi /r_o^l)\cdot(1/8B^2+2B^2(6+1+\rmin)^2-\rmin+r_o^u+1)~~~~~\\
\ge (\pi /r_o^l)\cdot(1/8B^2+2B^2(6+\rmin)^2-\rmin+r_o^u)~~~~~
\end{align*}

Therefore, we define~$m=B^{c_1}$ where~$c_1$ is a sufficiently large
constant. Note that we need to ensure that~$m\ge n$, which is however
no problem since we can, without loss of generality, assume that~$B$
is a multiple of~$n$ since we could simply multiply each input integer
as well as the bound by~$n$ to obtain a problem instance that is a
yes-instance if and only if the original instance was a yes instance
and whose size is polynomial in the size of the original
input.

For the approximate radii, the upper and lower bounds still
hold. Similar to the proof of
Lemma~\ref{lem:DTRstarNPH:UpperBoundary}, the upper bound of~$12 +
1/12$ for the separator distance provides the equation
$k = 2\sqrt{3} \cdot \sqrt{k+1} + (16+\frac{1}{12})$,
which has a solution for~$k=37.6 < 38 = r_o^u$. Since it is~$\roexact
\leq \roapprox$ and~$\rcexact \leq \rcapprox$, the lower bound~$r_o^l
= 6$ still holds.
 
Note that if the promised approximate values $\roapprox,\rcapprox$ for
$r_o$ and $r_c$ are used and the distance between the separators is
between $12$ and $12+1/4B^2$ the maximum distance between the bow and
the chord changes compared to the precise scenario. It holds now: $d
\leq
\rcapprox-(\sqrt{(\rcapprox+\rmin)^2-(6+1/8B^2+\rmin)^2}-\rmin)$. However,
the upper bound of $1/4B^2$ still holds true because of the
following.

\begin{align*}
\rcapprox-(\sqrt{(\rcapprox+\rmin)^2-(6+1/8B^2+\rmin)^2}-\rmin)\le 1/4B^2\Leftrightarrow \\
(\rcapprox+\rmin)-1/4B^2\le \sqrt{(\rcapprox+\rmin)^2-(6+1/8B^2+\rmin)^2}\Leftarrow \\
1/16B^4-(\rcapprox+\rmin)/2B^2\le -(6+1/8B^2+\rmin)^2\Leftrightarrow \\
1/8B^2+2B^2(6+1/8B^2+\rmin)^2-\rmin\le \rcapprox\Leftarrow \\
1/8B^2+2B^2(6+1/8B^2\rmin)^2-\rmin+\roapprox\le \roapprox/\mathrm{sin}(\pi /m)\Leftrightarrow \\
\mathrm{sin}(\pi /m)\le \roapprox/(1/8B^2+2B^2(6+1/8B^2+\rmin)^2-\rmin+\roapprox)\Leftarrow \\
\pi /m\le \roapprox/(1/8B^2+2B^2(6+1/8B^2+\rmin)^2-\rmin+\roapprox)\Leftrightarrow \\
m\ge (\pi /\roapprox)\cdot(1/8B^2+2B^2(6+1/8B^2+\rmin)^2-\rmin+\roapprox)\Leftarrow \\
m\ge (\pi /r_o^l)\cdot(1/8B^2+2B^2(6+1/8B^2+\rmin)^2-\rmin+r_o^u+1/B^{c_3})\Leftarrow \\
m\ge (\pi /r_o^l)\cdot(1/8B^2+2B^2(6+1+\rmin)^2-\rmin+r_o^u+1)~~~~~
\end{align*}

\begin{figure}[tb]
  \centering
  \includegraphics{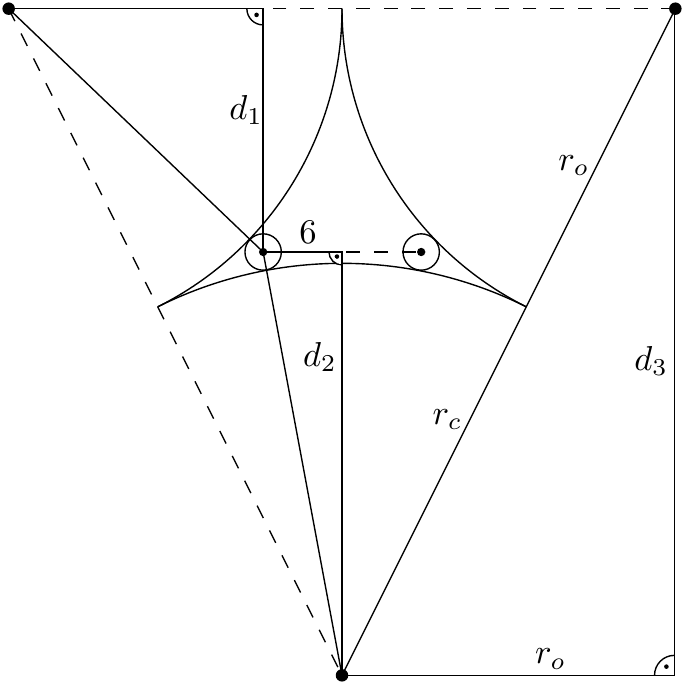}
  \caption{Determining the outer disk radius by utilizing the Pythagorean Theorem.}
  \label{fig:theom:DTRstarNPH:approx05}
\end{figure}

\textbf{Choosing the radii}.  For a tight packing of the outer disks
around the central disk, the radius of the central disk can be
described as~$r_c=r_o/\mathrm{sin}(\pi /m)-r_o$. By using the
Pythagorean Theorem, we obtain the following (see
Fig.~\ref{fig:theom:DTRstarNPH:approx05}).
\begin{align*}
\sqrt{(r_o/\sin(\pi/m))^2-r_o^2}=d_3=d_1+d_2=\\\sqrt{(r_o+\rmin)^2-(r_o-6-\rmin)^2}+\sqrt{(r_o/\sin(\pi/m)-r_o+\rmin)^2-(6+\rmin)^2}
\end{align*}
By keeping in mind that a lower bound for~$r_o$ is~$r_o^l=6$ we can determine a unique solution:
\begin{align*}
r_o=\frac{\rmin\sin(\pi/m)-3\rmin-6-2\sqrt{2\rmin^2+6\rmin-2\rmin^2\sin^2(\pi/m)-6\rmin\sin^2(\pi/m)}}{\sin(\pi/m)-1}
\end{align*}

Precise computation of this formula can take a superpolynomial amount
of time. We will show later how to compute suitable approximations
$\roapprox$ and~$\rcapprox$, such that $\roexact < \roapprox \le
\roexact +\eps_3$ and~$\roapprox /\sin(\pi /m)-\roapprox <\rcapprox
\le \roapprox /\sin(\pi /m)-\roapprox+\eps_4$.

First, let us consider a tight packing of the outer disks with
approximated outer disk radius~$\roapprox$, $\roexact < \roapprox \leq
\roexact + \eps_3$ and the corresponding precise central radius
$\rcexact = \roapprox /\sin(\pi /m)-\roapprox$. Note,
that for these radii the maximum possible distance between the separator disks in a gap increases to~$12 +
\eps_s$ for some~$\eps_s > 0$.

\begin{figure}[tb]
  \centering
  \includegraphics{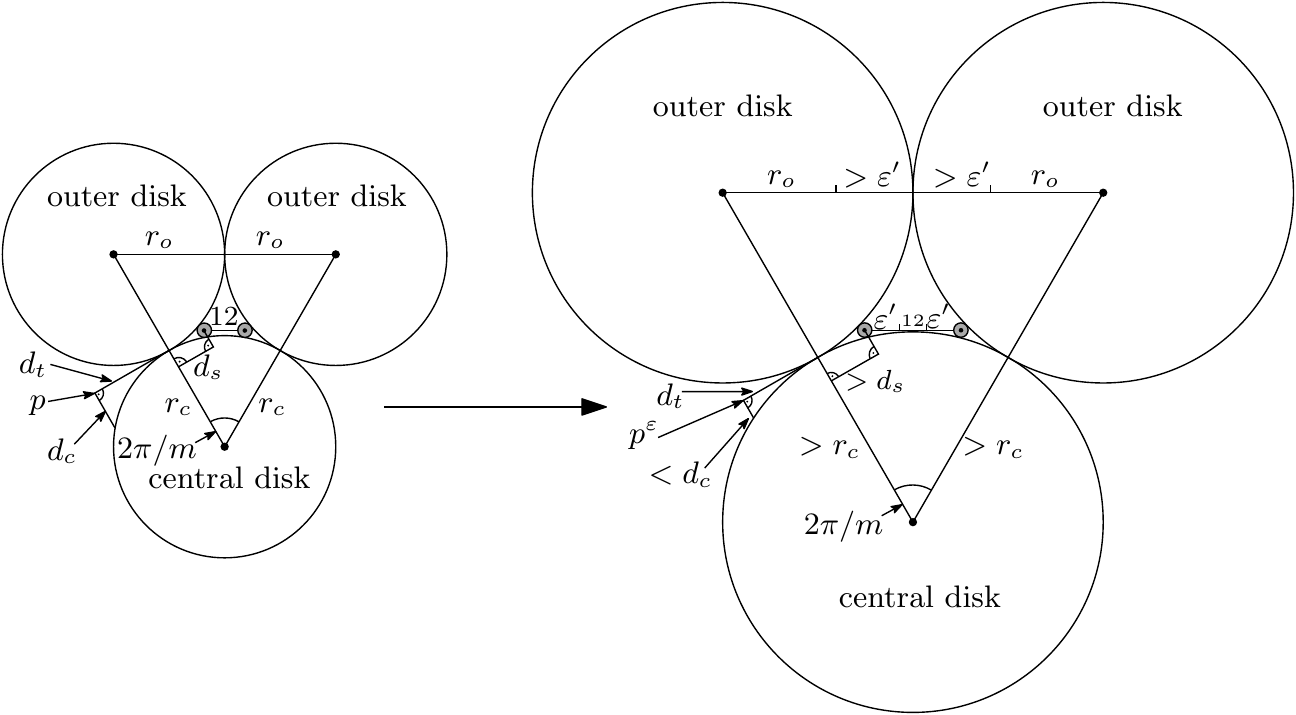}
  \caption{Increasing the separator distance by~$2\varepsilon'$
    increases the outer disks' radii by at least~$\varepsilon'$.}
  \label{fig:theom:DTRstarNPH:approx01}
\end{figure}

The left part of Fig.~\ref{fig:theom:DTRstarNPH:approx01}
illustrates two outer disks bounding a gap with the original radii and
the original separator distance of~$12$. Changing this distance
to~$12+2\varepsilon'$ (as depicted in the right part of
Fig.~\ref{fig:theom:DTRstarNPH:approx01}) increases the required
outer and central radii for a tight packing since we do not change~$m$
and, therefore, maintain the angle between the two outer disks. In
both packings, consider the tangent line between the respective left
outer disk and the central disk. We travel a distance~$d_t$ along
these tangent lines and arrive at points~$p$ and~$p^\varepsilon$ (see
Fig.~\ref{fig:theom:DTRstarNPH:approx01}). From these points we
travel orthogonally (to the tangents) until we reach the central
disk. Let~$d_c$ be the distance traveled in the original packing and
observe that the traveled distance in the modified packing is smaller
than~$d_c$ since the radius of the central disk is larger. On an
intuitive level, this means that the funnel-shaped regions next to the
tangent points become more narrow as the radii of the outer and
central disks increase. This phenomenon causes separator disks (which
maintain their original size) in the modified packing to be pushed
farther away from the lines that connect the centers of the central
and outer disks than in the original packing (the distance~$d_s$ in
Fig.~\ref{fig:theom:DTRstarNPH:approx01} increases). For this
reason, increasing the separator distance from~$12$
to~$12+2\varepsilon'$ pushes the centers of the outer disks at least
distance~$\varepsilon'$ to the sides since this is also the distance
that the separators move to the left or right respectively. We can
conclude that increasing the separator distance by~$2\varepsilon'$
increases the radius of the outer disks in a tight packing by at
least~$\varepsilon'$. The implication is that in a tight packing with
outer disk radius~$\roapprox \le \roexact+\eps_3$ and a corresponding
(precise) central radius~$\rcexact =\roapprox /\sin(\pi /m)-\roapprox$
the distance between the separators is at most~$12+2\varepsilon_3$.

\begin{figure}[tbp]
  \centering
   \includegraphics[width=0.45\textwidth]{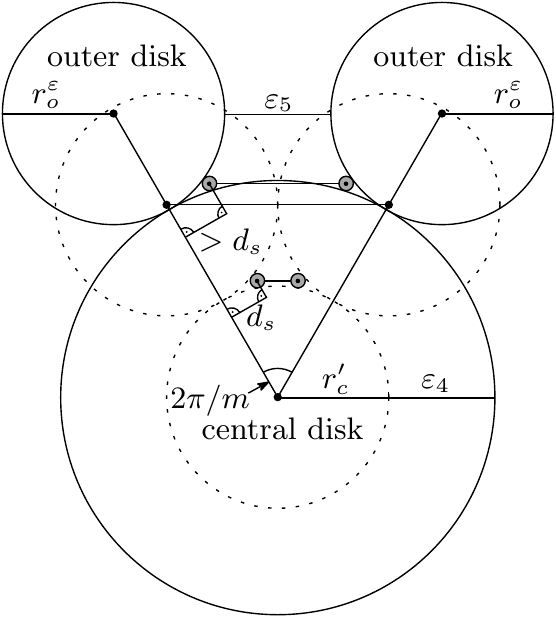}
   \caption{Increasing the central radius by~$\varepsilon_4$ creates a
     distance of~$\varepsilon_5$ between the outer disks. The distance
     between the separators increases by at most~$\varepsilon_5$.}
  \label{fig:theom:DTRstarNPH:approx02}
\end{figure}

Once again, we might be unable to compute the central
radius~$\rcexact$ precisely. Instead, we approximate it as~$\rcapprox$
with~$\rcexact < \rcapprox \le \rcexact+\eps_4= \roapprox /\sin(\pi
/m)-\roapprox+\eps_4$, which basically pushes the outer disks to the
outside as depicted in
Fig.~\ref{fig:theom:DTRstarNPH:approx02}. Assuming that the outer
disks can not deviate from these positions, this creates some
distance~$\varepsilon_5$ between the outer disks in each gap. The
outer disk radius remains~$\roapprox$ but the central disk radius is
larger than in a tight packing and has the value~$\rcapprox >
\rcexact$. Like in the argument in the previous paragraph, this causes
the separator disks to be pushed away from the lines that connect the
outer and central disks' centers. For this reason, the distance
between the separators increases by at most~$\varepsilon_5$ from at
most~$12+2\varepsilon_3$ to at most~$12+2\varepsilon_3+\varepsilon_5$.

\begin{figure}[htb]
  \centering
  \includegraphics{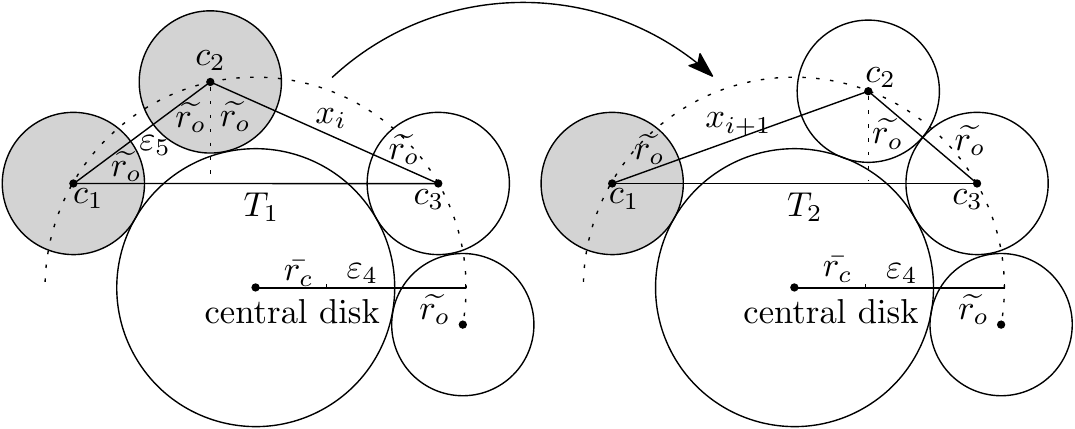}
  \caption{When allowing~$i+1$ instead of~$i$ outer disks  to move (non-movable outer disks in gray), the maximum distance increases at most linear in~$\varepsilon_5$.}
  \label{fig:theom:DTRstarNPH:approx03}
\end{figure}

So far, we have assumed that the outer disks can not deviate from
their positions even though they are placed distance~$\eps_5$ apart
from each other. In reality, however, the outer disks can rotate
around the central disk and, therefore, the distance between two outer
disks can increase to some value~$\eps_6>\eps_5$. We prove
that~$\eps_6<2m\eps_5$ by showing by induction that if we allow~$i$ of
the outer disks to move, the maximum distance~$x_i$ between two outer
disks is smaller than~$2(i+1)\varepsilon_5$ for any~$0\le i\le
m-1$. Clearly this holds true for~$x_0=\eps_5<2\eps_5$. Now assume
that our hypothesis is true for some fixed~$i\le m-2$. Clearly, the
distance~$x_i$ is maximized when we place all of the~$i$ movable disks
close together and thereby create one large gap. One of the
neighboring gaps is bounded by two non-movable (in step~$i$) disks
such that the distance between these disks is~$\eps_5$ as depicted in
the left part of Fig.~\ref{fig:theom:DTRstarNPH:approx03}. The
distance~$x_{i+1}$ gets maximized by now allowing the previously
non-movable disk next to the~$x_i$ gap to move such that the two gaps
merge as illustrated in the right part of
Fig.~\ref{fig:theom:DTRstarNPH:approx03}. Consider the
triangles~$T_1$ (left) and~$T_2$ (right) formed by the
centers~$c_1,c_2,c_3$ in
Fig.~\ref{fig:theom:DTRstarNPH:approx03}. The base side of these
triangles is identical, the height of~$T_2$ is smaller than the height
of~$T_1$ and the circumcircle of both triangles has
radius~$\rcexact+\eps_4+\roapprox$. We can conclude that the area
of~$T_2$ is smaller than the area of~$T_1$. Recalling that the area of
a triangle~$T$ with sides~$a,b,c$ can be described as~$abc/(4r)$
where~$r$ is the radius of the circumcircle of~$T$, we
obtain~$2\roapprox (2\roapprox +x_{i+1})<(2\roapprox
+\eps_5)(2\roapprox +x_i)$ and,
hence,~$x_{i+1}<x_i+\eps_5+\eps_5x_i/(2\roapprox
)<2(i+1)\eps_5+\eps_5+\eps_5(2(i+1)\eps_5)/(2\roapprox )$ by our
induction hypothesis. By choosing~$\eps_4$ and, therefore,~$\eps_5$
such that~$2m\eps_5<1$ we obtain
that~$x_{i+1}<2(i+1)\eps_5+\eps_5+\eps_5/(2\roapprox
)<2(i+1)\eps_5+\eps_5+\eps_5=2(i+2)\eps_5$, which concludes our
induction proof. Thus, the maximum distance between two outer disks
increases to at most~$\eps_6<2m\eps_5$. This increases the maximum
distance between the separators to at most~$12+2\eps_3+2m\eps_5$ as
illustrated in Fig.~\ref{fig:theom:DTRstarNPH:approx04}.

\begin{figure}[tbp]
  \centering
   \includegraphics[width=0.5\textwidth]{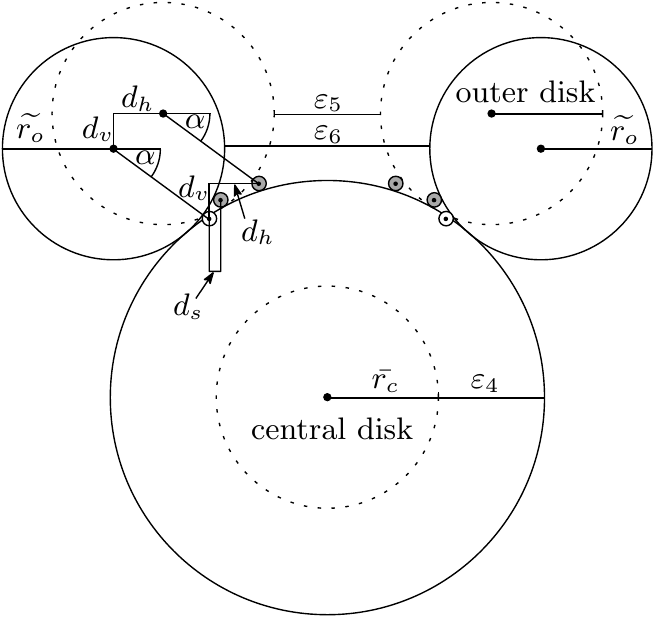}
   \caption{Increasing the distance between two outer disks
     from~$\varepsilon_5$ to~$\varepsilon_6$ increases the distance
     between the separators by at most~$\varepsilon_6-\varepsilon_5$
     since~$d_s\ge 0$.}
  \label{fig:theom:DTRstarNPH:approx04}
\end{figure}

With basic trigonometry, we determine
that~$2\roapprox+\eps_5=2(\rcexact+\eps_4+\roapprox)\sin(\pi/m)$
and~$2\roapprox=2(\rcexact+\roapprox)\sin(\pi/m)$, see
Fig.~\ref{fig:theom:DTRstarNPH:approx02}. We combine these two
equalities and obtain~$\eps_5=2\eps_4\sin(\pi/m)<2\eps_4\pi/m<2\eps_4
\cdot 4/m$. The maximum distance between two separators is, therefore,
at most~$12+2\eps_3+16\eps_4$. 

Recall that by Lemma~\ref{lem:DTRstarNPH:nonTwelve},
Properties~\ref{pr:DTRstarNPH:fitIffFeasible}
and~\ref{pr:DTRstarNPH:feasibleOnlyWithSep} hold true for our radius
function~$r$ as long as the maximum distance between two separators is
at most~$12+\eps_s$ with~$\eps_s=1/4B^2$. We can now simply
choose~$\eps_3=1/B^{c_3}$ and~$\eps_4=1/B^{c_4}$ such
that~$2\eps_3+16\eps_4 \le \eps_s$. Therefore, we can conclude that
the approximate radii for the outer and central disks suffice.

\textbf{Approximating radii in polynomial time}. It remains to argue
that we can approximate our radii as required. The formulas for the
exact radii~$\roexact$ and~$\roexact$ contain a constant number of
square root and sine operations. Recall that~$m=B^{c_1}$. Redefining
and increasing~$m$ such that~$m=2^p$ with~$2^{p-1}<B^{c_m}\le 2^p=m$
causes no issues for our construction. Therefore, by using the half-angle
formula

\begin{equation*}
  \cos \left(\frac{1}{2}x\right) = \sqrt{\frac{1+\cos x}{2}} \textnormal{ for } 0 < x < \pi,
\end{equation*}
 
and using~$\sin x = \sqrt{1 - \cos^2 x}$ for $0 < x < \pi/2$, we can
replace each sine operation in our formulas by~$p=\log_2 m$ nested
square root operations. In total, we therefore perform~$O(\log
m)=O(\log B^{c_1})$ square root operations. Individually, each square
root approximation can be performed in polynomial time using Heron's
quadratically converging method since we can easily determine constant
upper and lower bounds for each square root term and use these as the
initiation values. In order to approximate the nested square roots, we
need to increase the approximation accuracy by an according polynomial
amount. \qednew
\end{proof}

Lemma~\ref{lem:DTRstarNPH:addGaps} already showed how to construct an
equivalent 3-Partition instance with~$3m \geq 3n$ input
integers.
We now have all the tools required to prove the main result of this
section. Lemmas~\ref{lem:DTRstarNPH:addGaps}
and~\ref{lem:DTRstarNPH:chooseRad} show that the construction can be
performed in polynomial
time. Properties~\ref{pr:DTRstarNPH:fitIffFeasible}
and~\ref{pr:DTRstarNPH:feasibleOnlyWithSep} let us show that a valid
distribution of the input and separator disks among the gaps induces a
solution of the 3-Partition instance and vice versa.

\newcommand{\ThmFrStarsHardText}{The WDC graph recognition problem is
  (strongly) \NP-hard even for stars if an arbitrary embedding is allowed.}
\begin{theorem}
  \label{thm:fr-stars-hard}
  \ThmFrStarsHardText
\end{theorem}
\begin{proof}
  Given a 3-Partition instance~$(\A,B)$, an equivalent
  instance~$(\A',B')$ as in Lemma~\ref{lem:DTRstarNPH:addGaps} and an
  equivalent WDC graph recognition instance can be constructed in
  polynomial time. The number of disks is linear in~$m$ and, thus,
  polynomial in~$B$. For the input and separator disks the radius
  computation does not cause any complications since the output of our
  radius function~$r$ is always a polynomially bounded rational
  number. For the inner and outer disks, the radii can be approximated
  in polynomial time; see
  Lemma~\ref{lem:DTRstarNPH:chooseRad}. Furthermore, the encoding size
  of the WDC graph recognition instance is polynomial in the encoding
  size of~$(\A,B)$. A solution of $(\A',B')$ induces a valid
  distribution of disks among the $m$ gaps by placing each disk triple
  together with two separators in each gap.
  Conversely, a valid distribution of the input and separator disks
  among the~$m$ gaps induces a solution of $(\A',B')$,
  since~Properties~\ref{pr:DTRstarNPH:fitIffFeasible}
  and~\ref{pr:DTRstarNPH:feasibleOnlyWithSep} ensure that each of
  the~$m$ gaps contains a feasible triple and two separators. \qednew
\end{proof}

\subsection{Recognizing embedded stars with a weighted disk contact representation}\label{sec:embeddedstars}

If, however, the order of the leaves around the central vertex of the
star is fixed, the existence of a WDC representation can be decided by
iteratively placing the outer disks $D_1$, \ldots, $D_{n-1}$ tightly
around the central
disk~$D_c$. %
A naive approach tests for collisions with all previously added disks and yields a total runtime of $O(n^2)$. However, in the following theorem we improve this to $O(n)$ by maintaining a list containing only disks that might be relevant in the future.

\newcommand{\ThmFrStarFeAlgoText}{On a Real RAM, for an embedded, vertex-weighted star $S$ it can be decided in linear time whether $S$ is a WDC graph. A WDC representation respecting the embedding (if one exists) can be constructed in linear time.}
\begin{theorem}
  \label{thm:fr-stars-fe-algo}
  \ThmFrStarFeAlgoText
\end{theorem}
\begin{proof}
Let $r_i$ be the radius of~$D_i$, and assume that~$D_1$ is the largest
outer disk. Then, $D_2$ can be placed next to $D_1$ clockwise. Suppose
we have already added $D_2$, \ldots, $D_i$. As depicted in
Fig.~\ref{fig:fr-stars-fe-algo}, tightly placing $D_{i+1}$ next
to~$D_i$ might cause $D_{i+1}$ to intersect with a disk inserted
earlier, even with~$D_1$. Testing for collisions with all previously
added disks yields a total runtime of $O(n^2)$; we improve this
to~$O(n)$ by keeping a list~$L$ of all inserted disks that might be
relevant for future insertions. Initially, only~$D_1$ is in~$L$. We
shall see that~$L$ remains sorted by non-increasing radius.

\begin{figure}[t]
  \centering
  \includegraphics[width=0.45\textwidth]{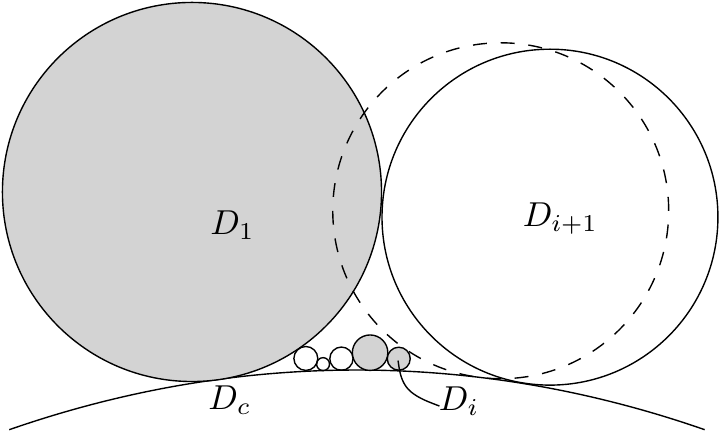}
  \caption{Deciding existence for
    Theorem~\ref{thm:fr-stars-fe-algo}. Gray disks are in~$L$ before
    inserting~$D_{i+1}$. After that, the two small gray disks will be
    removed from~$L$. }
  \label{fig:fr-stars-fe-algo}
\end{figure}

When inserting~$D_{i+1}$, we traverse~$L$ backwards and test for
collisions with traversed disks, until we find the largest index $j<i$
such that~$r_j \in L$ and~$r_{i+1} \leq r_j$. Next, we place~$D_{i+1}$
tightly next to all inserted disks, avoiding collisions with the
traversed disks.

First, note that~$D_{i+1}$ cannot intersect disks preceding~$D_j$
in~$L$ (unless $D_{i+1}$ and~$D_1$ would intersect clockwise, in which
case we report non-existence). Next, disks that currently
succeed~$D_j$ in~$L$ will not be able to intersect~$D_{i+2}$, \ldots,
$D_{n-1}$ and are therefore removed from~$L$. Finally, we
add~$D_{i+1}$ to the end of~$L$. Since all but one traversed disks are
removed during each insertion, the total runtime is~$O(n)$. We return
the constructed WDC representation if we can insert all disks tightly and there is still space left; otherwise we report non-existence.\qednew
\end{proof}

{\small 
   \bibliographystyle{titto-lncs-01}
   \bibliography{abbrv,bib}

}

\end{document}